\documentclass[12pt]{article}
\usepackage[a4paper,left=1in,top=1in,right=.7in,bottom=1in]{geometry}

\usepackage[utf8]{inputenc}
\usepackage{algorithm}
\usepackage[noend]{algpseudocode}
\usepackage{enumitem}
\usepackage{amsmath,amssymb,amsfonts,amsthm,latexsym}
\usepackage{mathrsfs}
\usepackage{graphicx}
\usepackage{color}
\usepackage{mathtools}
\usepackage[normalem]{ulem}
\usepackage{scalerel}
\usepackage[mathcal]{euscript}
\usepackage{algorithm}
\usepackage[noend]{algpseudocode}
\usepackage{xspace}
\usepackage{ifthen}

\usepackage{url}

\usepackage{authblk}
\usepackage[numbers, sort&compress]{natbib}

\usepackage[colorlinks]{hyperref}
\hypersetup{
	citecolor=blue,
   linkcolor=black,
}
\usepackage[font=scriptsize,width=.85\textwidth,labelfont=bf]{caption}
\usepackage{mathptmx}

{\begin{description}[leftmargin = 0.2cm, labelsep = 0.2cm]}
  {\end{description}}

\newcommand\doverline[1]{\ThisStyle{%
  \setbox0=\hbox{$\SavedStyle\overline{#1}$}%
  \ht0=\dimexpr\ht0-.15ex\relax% CHANGE .15 TO AFFECT SPACING
  \overline{\copy0}%
}}

\newcommand{\mc}{\mathcal}
\newcommand{\ms}{\mathscr}
\newcommand{\wh}{\widehat}
\providecommand{\keywords}[1]{\textbf{\textit{Keywords: }} #1}
\newcommand{\lc}[1]{\ensuremath{_{{\scriptscriptstyle #1}}}}

\newcommand{\nameT}{$\eps$-tree\xspace}
\newcommand{\Pcol}{$P$-recolored\xspace}

\newcommand{\X}{\ensuremath{X}\xspace}
\newcommand{\irr}[1]{\ensuremath{[#1]_{\textrm{irr}}}\xspace}
\newcommand{\PS}[1]{\ensuremath{\mc{P}(#1)}\xspace}
\newcommand{\eps}{\ensuremath{\varepsilon}\xspace}

\newcommand{\eXM}{\ensuremath{\eps \colon \irr{\X \times \X} \to \PS{M}}\xspace}

\newcommand{\EC}{\textsf{ELC}\xspace}
\newcommand{\HC}{\textsf{HLC}\xspace}
\newcommand{\IC}{\textsf{IC}\xspace}

\newcommand{\NnotCol}[1]{\ensuremath{{N}_{\neg #1}}\xspace}
\newcommand{\Ns}{\ensuremath{\mc{N}}\xspace}
\newcommand{\To}[1]{\ensuremath{\ms{T}({#1})}\xspace}

\newcommand{\Ts}{\ensuremath{T^*}\xspace}
\newcommand{\ls}{\ensuremath{\lambda^*}\xspace}
\newcommand{\Cls}{\ensuremath{\mc{C}}\xspace}
\newcommand{\Cl}[1]{\ensuremath{C_{\scriptscriptstyle #1}}\xspace}

\newcommand{\lca}{\ensuremath{\operatorname{\mathrm{lca}}}}
\newcommand{\parent}{\mathrm{par}}
\newcommand{\rootT}[1]{\ensuremath{\rho_{_{{\scriptscriptstyle #1}}}}}

\usepackage[nameinlink]{cleveref}

\crefname{theorem}{Thm.}{Thm.}
\Crefname{theorem}{Theorem}{Theorems}
\crefname{lemma}{L.}{L.}
\Crefname{lemma}{Lemma}{Lemmas}
\crefname{proposition}{Prop.}{Prop.}
\Crefname{proposition}{Proposition}{Propositions}
\crefname{corollary}{Cor.}{Cor.}
\Crefname{corollary}{Corollary}{Corollaries}
\crefname{definition}{Def.}{Def.}
\Crefname{definition}{Definition}{Definitions}
\crefname{figure}{Fig.}{Fig.}
\Crefname{figure}{Figure}{Figures}
\crefname{algorithm}{Alg.}{Alg.}
\Crefname{algorithm}{Algorithm}{Algorithms}

\newtheorem{theorem}{Theorem}[section]
\newtheorem{lemma}[theorem]{Lemma}
\newtheorem{proposition}[theorem]{Proposition}
\newtheorem{corollary}[theorem]{Corollary}

\theoremstyle{definition}
\newtheorem{exmpl}[theorem]{Example}
\newtheorem{definition}[theorem]{Definition}

\providecommand{\keywords}[1]{\textbf{\textit{Keywords: }} #1}

\makeindex             % used for the subject indexs

\title{Generalized Fitch Graphs II: Sets of Binary Relations that are explained by Edge-labeled Trees}
\author[1,2]{Marc Hellmuth} 
\author[3,4]{Carsten R.\ Seemann}
\author[3,4,5,6,7,8]{Peter F.\ Stadler}

\affil[1]{Institute	 of Mathematics and Computer Science, University of Greifswald, Walther-
  Rathenau-Strasse 47, D-17487 Greifswald, Germany  \\ 	
	Email: \texttt{mhellmuth@mailbox.org}}

\affil[2]{
	Saarland University, Center for Bioinformatics, Building E 2.1, P.O.\ Box 151150, D-66041 Saarbr{\"u}cken, Germany
	  }

\affil[3]{Max Planck Institute for Mathematics in the Sciences,
  Inselstra{\ss}e 22, D-04103 Leipzig, Germany}

\affil[4]{Bioinformatics Group, Department of Computer Science and
  Interdisciplinary Center for Bioinformatics, University of Leipzig,
  H{\"a}rtelstra{\ss}e 16-18, D-04107 Leipzig, Germany}

\affil[5]{Bioinformatics Group, Department of Computer Science;
  Interdisciplinary Center for Bioinformatics; German Centre for
  Integrative Biodiversity Research (iDiv) Halle-Jena-Leipzig; Competence
  Center for Scalable Data Services and Solutions Dresden-Leipzig; Leipzig
  Research Center for Civilization Diseases; and Centre for Biotechnology
  and Biomedicine, University of Leipzig, H{\"a}rtelstra{\ss}e 16-18,
  D-04107 Leipzig, Germany}

\affil[6]{Institute for Theoretical Chemistry, University of Vienna,
  W{\"a}hringerstra{\ss}e 17, A-1090 Wien, Austria}

\affil[7]{Facultad de Ciencias, Universidad Nacional de Colombia,
  Sede Bogot{\'a}, Colombia}

\affil[8]{The Santa Fe Institute, 1399 Hyde Park Rd., Santa Fe, NM
  87501, United States}

\begin{document}

\date{\ }

\maketitle

\abstract{ 
Fitch graphs $G=(X,E)$ are digraphs that are explained by $\{\emptyset, 1\}$-edge-labeled rooted trees $T$ with leaf set $X$: there is an arc $(x,y) \in E$ if and only if the unique path in $T$ that connects the last common ancestor $\lca(x,y)$ of $x$ and $y$ with $y$ contains at least one edge with label ``1’’. In practice, Fitch graphs represent xenology relations, i.e., pairs of genes $x$ and $y$ for which a horizontal gene transfer happened along the path from $\lca(x,y)$ to $y$.

In this contribution, we generalize the concept of Fitch graphs and consider trees $T$ that are equipped with edge-labeling $\lambda\colon E\to \mc{P}(M)$ that assigns to each edge a subset $M'\subseteq M$ of colors.  Given such a tree, we can derive a map $\eps\lc{(T,\lambda)}$ (or equivalently a set of not necessarily disjoint binary relations), such that $i\in \eps\lc{(T,\lambda)}(x,y)$ (or equivalently $(x,y)\in R_i$) with $x,y\in X$, if and only if there is at least one edge with color $i$ from $\lca(x,y)$ to $y$. 

The central question considered here: Is a given map $\eps$ a Fitch map, i.e., is there there an edge-labeled tree $(T,\lambda)$ with $\eps\lc{(T,\lambda)} = \eps$, and thus explains $\eps$? Here, we provide a characterization of Fitch maps in terms of certain neighborhoods and forbidden submaps.  Further restrictions of Fitch maps are considered. Moreover, we show that the least-resolved tree explaining a Fitch map is unique (up to isomorphism). In addition, we provide a polynomial-time algorithm to decide whether $\eps$ is a Fitch map and, in the affirmative case, to construct the (up to isomorphism) unique least-resolved tree $(\Ts,\ls)$ that explains $\eps$.
}
\smallskip

\noindent

\sloppy

\section{Introduction}
\label{sec:intro}

Labeled rooted trees arise naturally as models of evolutionary processes in
mathematical biology. Both vertex and edge labels are used to annotate
classes of evolutionary events. The connection to both empirically
accessible data and to key biological concepts is given by binary relations
on the leaves that are specified in terms of the labels encountered in
certain substructures within the underlying tree. The practical importance
of this type of models derives from the fact that the relations can be
inferred directly from empirical data, such as gene sequences, without
knowledge of the tree \cite{Lechner:11a,Nichio:17,Ravenhall:15}. From a
mathematical perspective, a rich set of interrelated graph-theoretical
problems arises from the questions which relations on the leaves can be
obtained from labeled trees under a given set of rules.

Relations and edge-labeled graphs defined in terms of
  \emph{vertex}-labeled trees have been widely studied since the 1970's and
  range from cographs \cite{Jung:78,Corneil:81,Hellmuth:13a,Lafond:14} and
  di-cographs \cite{Crespelle:06} to 2-structures
  \cite{ER1:90,ER2:90,EHPR:96}, symbolic ultrametrics
  \cite{Boecker:98,Hellmuth:17a} or three-way symbolic tree-maps
  \cite{Huber2018a,Gruenewald2018}. In contrast, relations and edge-labeled
  graphs that are defined in terms of \emph{edge}-labeled trees have just
  been explored recently. Edge-labels may represent the number of events,
in which case they lead to pairwise compatibility graphs (PCGs) and their
variants: here, an edge is drawn if the total weight along the path
connecting $x$ and $y$ lies between \emph{a priori} defined bounds
\cite{PCGsurvey}. Leaf power graphs specify either only an upper or a lower
bound \cite{Fellows:08}.  While PCGs are defined with strictly positive
edge weights, an extension to zero weights -- the absence of evolutionary
events along an edge -- is required in models of evolution focused on rare
events \cite{Hellmuth:18b}.

\begin{figure}
  \begin{center}
    \includegraphics[width=0.95\textwidth]{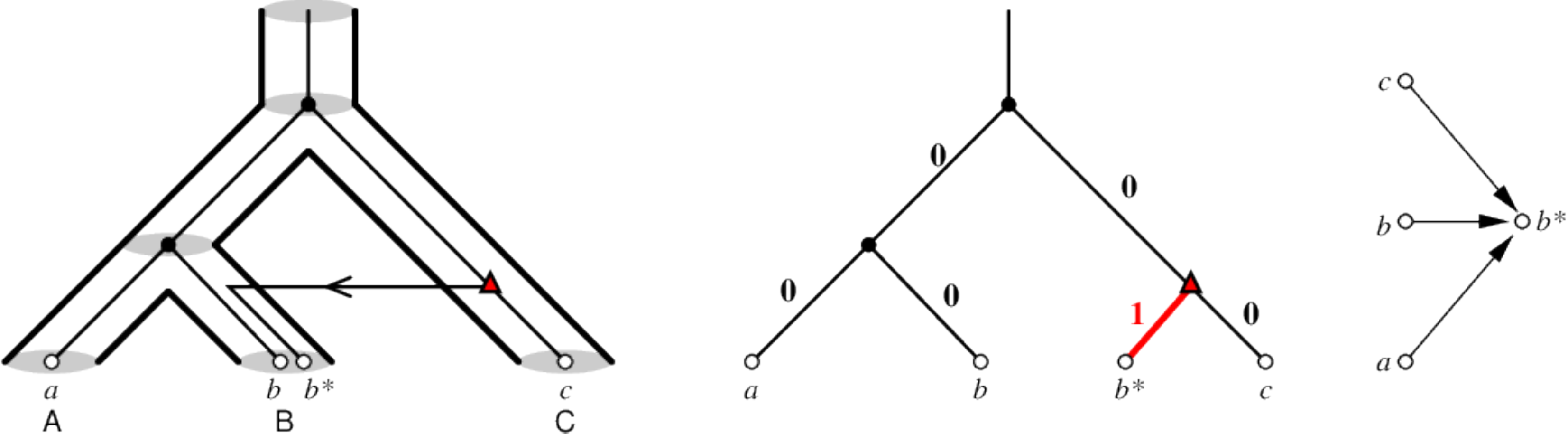}
  \end{center} 
  \caption{The evolution of gene families is modeled as an embedding
      of a gene tree (thin lines, and r.h.s.\ panel, with four genes $a$,
      $b$, $b^*$ and $c$) into a species tree (shown as tubes with fat
      outlines with three species $\mathsf{A}$, $\mathsf{B}$, and
      $\mathsf{C}$). At each speciation (gray ellipses), each gene
        present in the genome is transmitted into both descending
      lineages, corresponding to a speciation even in the gene tree (shown
      as $\bullet$).  Horizontal gene transfer consists in a duplication of
      a gene, one copy of which ``jumps'' into a different lineage. The
      corresponding edge in the gene tree is marked with label $1$. The
      corresponding Fitch digraph has an edge from $x$ to $y$ if there is
      at least one HGT event on the path from the last common ancestor
      $\lca(x,y)$ and $y$ in the gene tree.}
  \label{fig:intro}
\end{figure}

Fitch graphs were introduced to model so-called horizontal transfer events
\cite{Geiss:18a} based on the seminal work by Walter M.\ Fitch
\cite{Fitch:00}, see Fig.~\ref{fig:intro} for an illustration of
the model.  Fitch graphs can be seen as directed generalizations of
lower bound leaf power graphs: a directed edge connects $x$ and $y$ if at
least one of the tree edges connecting the last common ancestor $\lca(x,y)$
and the ``target'' leaf $y$ carries a ``horizontal transfer'' label
\cite{Geiss:18a}.  Modeling different types of events by different labels
yields a multi-colored generalization of Fitch graphs that can be regarded
as a collection of edge-disjoint sets of Fitch graphs
\cite{Hellmuth:2019d}. The colors can be used e.g.\ to distinguish
  genomic locations where the horizontally transferred gene copy is inserted,
  and adds to the information that can be extracted for the colored Fitch
  graphs compared to their color-free version. Here, we further relax the
compatibility conditions and consider sets of Fitch maps and trees whose
edges are labeled by finite sets. Conceptually, the construction explored
in this contribution can be seen as an edge-centered analog of the
2-structures explored in \cite{ER1:90,ER2:90,Hellmuth:17a}.

An uncolored Fitch graph is explained by an unique least-resolved
  trees which, in the context of gene families, is obtained by a series of
  edge-contractions from the true gene tree. Fitch graphs therefore encode
  constraints on the evolutionary history. From a mathematical point of
  view, Fitch graphs are a subclass of the directed cographs
  \cite{Crespelle:06} characterized by a small set of forbidden induced
  subgraphs \cite{Geiss:18a}. An alternative characterization
  \cite{Hellmuth:19F} makes use of certain neighborhood systems that will
  also play a key role here.
  
This contribution is organized as follows. In \Cref{sec:prelim}, we provide
most of the necessary definitions needed here, and continue to characterize
Fitch maps in \Cref{sec:char}.  To this end, we introduce the notion of
(complementary) neighborhoods that additionally provide the necessary
information to reconstruct a tree that explains a given Fitch map.  In
addition, we provide a so-called inequality-condition that is needed to
obtain a correct edge-labeling of the underlying trees.  In
\Cref{sec:uniq}, we utilize the latter results and show that every Fitch
map is explained by a (up to isomorphism) unique least-resolved tree.
Moreover, we show that every Fitch map is characterized in terms of
so-called forbidden submaps.  In \Cref{sec:simple}, we consider
$k$-restricted Fitch maps, i.e., Fitch maps that are explained by
edge-labeled trees for which the number of colors on their edges does not
exceed a prescribed integer $k$.  We provide a constructive
characterization of $k$-restricted Fitch maps, and show that, in general,
$k$-restricted Fitch maps cannot be characterized in terms of forbidden
submaps. In \Cref{sec:algo}, we finally provide a polynomial-time algorithm
to recognize Fitch maps $\eps$ and, in the affirmative case, to reconstruct
the unique least-resolved tree that explains $\eps$.  We complete this work
with a short outlook where we provide a couple of open questions for
further research.

\section{Preliminaries}
\label{sec:prelim}

\paragraph{Basics}
For a finite set $\X$ we put
$\irr{\X \times \X} \coloneqq \X \times \X \setminus \{ (x,x) : x \in \X
\}$, and $\binom{\X}{k} \coloneqq \{ \X' \subseteq \X : |\X'|=k \}$.
The \emph{power set} $\PS{\X}$ of $\X$ comprises all subsets of
$\X$. In the following, we consider \emph{maps} $f\colon X\to Y$ that
associate to every element of the set $X$ exactly one element of the set
$Y$. Moreover, we consider \emph{(undirected) graphs}, resp., \emph{di-graphs} $G=(V,E)$ with finite vertex set $V$ and edge
set $E\subseteq \binom{V}{2}$, resp.,  arc set $E\subseteq \irr{\X\times\X}$.  
Hence, the graphs considered here do not
contain loops or multiple edges. A graph $H=(W,F)$ is a \emph{subgraph} of
$G=(V,E)$, denoted by $H\subseteq G$, if $W\subseteq V$ and $F\subseteq E$.

\paragraph{Trees}
A \emph{rooted tree} is a connected, cycle-free graph with a distinguished
vertex $\rootT{T} \in V$, called the \emph{root} of $T$.  Let
$T=(V,E)$ be a rooted tree.  Then, the unique path from the vertex
$v \in V$ to the vertex $w\in V$ is denoted by $P\lc{T}(v,w)$. A
\emph{leaf} of $T$ is a vertex $v\in V\setminus\{\rootT{T}\}$ such that
$\deg\lc{T}(v)=1$. The set of all leaves of $T$ will be denoted by
$\mc{L}(T)$.  The vertices in
$\mathring{V}(T)\coloneqq V\setminus\mc{L}(T)$ are called \emph{inner}
vertices.  All edges in
$\mathring{E}(T)\coloneqq\{ \{v,w\}\in E: v,w \in \mathring{V}(T) \}$ are
called \emph{inner} edges. Edges of $T$ that are not contained in
$\mathring{E}(T)$ are called \emph{outer} edges.  Every rooted tree carries
a natural partial order $\preceq\lc{T}$ on the vertex set $V$ that can be
obtained by setting $v \preceq\lc{T} w$ if and only if the path from
$\rootT{T}$ to $w$ contains $v$. In this case, we call $v$ an
\emph{ancestor} of $w$, $w$ a \emph{descendant} of $v$, and say that $v$
and $w$ are \emph{comparable}. Instead of writing $v \preceq\lc{T} w$ and
$v \ne w$, we will use $v \prec\lc{T} w$.

It will be convenient to use a notation for edges $\{v,w\} \in E$ that
implies which one of the vertices in $\{v,w\}$ is closer to the
root. Therefore, we always write $(v,w) \in E$ to indicate that
$v \prec\lc{T}w$. In this case, the unique vertex $v$ is called
\emph{parent} of $w$, denoted by $\parent\lc{T}(w)$.  For a non-empty
subset $V'\subseteq V$ of vertices, the \emph{last common ancestor of
  $V'$}, denoted by $\lca\lc{T}(V')$, is the unique
$\preceq\lc{T}$-maximal.  vertex of $T$ that is an ancestor of every vertex
in $V'$.  We will make use of the simplified notation
$\lca\lc{T}(x,y)\coloneqq\lca\lc{T}(\{x,y\})$ for $V'=\{x,y\}$.  We will
omit the explicit reference to $T$ for $\preceq\lc{T}$, $\parent\lc{T}(w)$
and $\lca\lc{T}$, whenever it is clear which tree is considered.

A \emph{phylogenetic tree $T$ on $\X$} is a rooted tree $T$ with leaf set
$\mc{L}(T)=\X$, with the degree $\deg\lc{T}(\rootT{T})\ge 2$, and the
degree $\deg\lc{T}(v)\ge3$ for every inner vertex
$v \in \mathring{V}(T)\setminus\{\rootT{T}\}$.

A \emph{rooted triple}, denoted by $xy|z$, is a phylogenetic tree on
$\{x,y,z\}$ with $\lca(x,y,z) \prec \lca(x,y)$. A triple $xy|z$ is
\emph{displayed} by a rooted tree $T$, if
$\lca\lc{T}(x,y,z) \prec\lc{T} \lca\lc{T}(x,y)$. We denote with
$\mathfrak{R}(T)$ the set of all triples that are displayed by $T$.  A set
$R$ of triples is called \emph{compatible} if
$\langle R \rangle\ne \emptyset$, where $\langle R \rangle$ denotes the set
of all trees that display $R$.  In other words, $R$ is compatible if
  there is a tree $T$ with $R\subseteq \mathfrak{R}(T)$, see
  \Cref{fig:exmp-triples} for an illustrative example.  Moreover, the
closure $\overline R$ of an arbitrary compatible set $R$ of triples is
defined by $\overline R = \bigcap_{T\in \langle R\rangle} \mathfrak{R}(T)$.
In other words, $\overline R$ contains all triples that are displayed by
every tree that also display $R$, see e.g.\ \cite{BS:95,GSS:07,HS:17} for
further details. In fact, $\overline R$ satisfies the usual properties for
a closure operator \cite{BS:95}, i.e., $R \subseteq \overline R$,
$\doverline{R}= \overline R$, and if $R' \subseteq R$, then
$\overline{R'}\subseteq \overline R$.

\begin{figure}[tbp]
  \begin{center}
    \includegraphics[width=.5\textwidth]{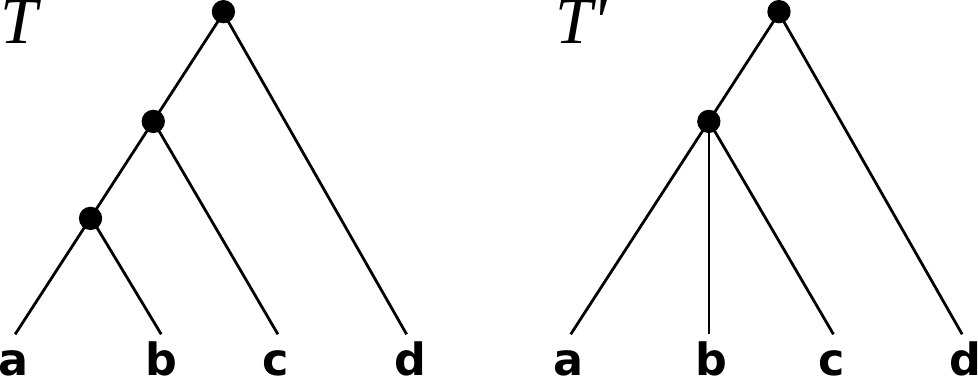}
  \end{center}
  \caption{ Shown are two (phylogenetic) trees $T$ and $T'$ on
      $\X = \{a,b,c,d\}$ that display the set $R = \{ac|d,bc|d\}$ of rooted
      triples.  For the set $R' = \{ab|c,ac|d,bc|d\}$ there is only the
      tree $T$ that displays $R'$.  Thus,
      $\overline{R'} = \mathfrak{R}(T) = \{ab|c,ac|d,bc|d,ab|d\}$.  In
      particular, $ab|c$ is not displayed by $T'$.  In this example,
      $\Cls(T) = \{\X, \{a,b,c\},\{a,b\}, \{a\},\{b\},\{c\},\{d\}\}$ and
      $\Cls(T') = \Cls(T) \setminus \{\{a,b\}\}$. }
  \label{fig:exmp-triples}
\end{figure}

\paragraph{Clusters and Hierarchies}
Let $\X$ be a finite set, and let $\mc{H}\subseteq \PS{\X}$ be a set system
on $\X$. Then, we say that $\mc{H}$ is \emph{hierarchy-like} if
$P\cap Q \in \{P,Q,\emptyset\}$ for all $P,Q \in \mc{H}$. The set system
$\mc{H}$ is a \emph{hierarchy (on $\X$)} if it is hierarchy-like and in
addition satisfies $\X\in\mc{H}$ and $\{x\}\in\mc{H}$ for all $x\in\X$.
  
Given a phylogenetic tree $T=(V,E)$, we define for each vertex $v \in V$
the set of descendant leaves as
$\Cl{T}(v)\coloneqq \{ x \in \mc{L}(T): v \preceq\lc{T} x \}$. We say that
$\Cl{T}(v)$ is a \emph{cluster} of $T$. Moreover, the \emph{cluster set of
  $T$} is $\Cls(T) \coloneqq \{ \Cl{T}(v) : v \in V \}$. In this context,
it is well-known that $\Cls(T)$ forms a hierarchy and that there is a
one-to-one correspondence between (isomorphism classes of) rooted trees and
their cluster sets:
\begin{lemma}[{\cite[Thm.\ 3.5.2]{Semple2003}}]  \label{lem:tree-iff-cluster}
  For a given subset $H\subseteq \PS{\X}$, there is a phylogenetic tree $T$
  on $\X$ with $H=\Cls(T)$ if and only if $H$ is a hierarchy on
  $\X$. Moreover, if there is such a phylogenetic tree $T$ on $\X$, then,
  up to isomorphism, $T$ is unique.
\end{lemma}

\section{Characterization of Generalized Fitch maps}
\label{sec:char}

\subsection{Definitions}

\begin{definition}
  Let $M$ be an arbitrary finite set of colors. An \emph{edge-labeled
    (phylogenetic) tree $(T,\lambda)$ on $\X$ (with $M$)} is a phylogenetic
  tree $T=(V,E)$ on $\X$ together with a map $\lambda: E \to \PS{M}$
  that assigns to every edge $e \in E$ exactly one subset
    $\lambda(e)\subseteq M$ of colors.
\end{definition}
We will often refer to the map $\lambda$ as the \emph{edge-labeling} and
call $e$ an \emph{$m$-edge} if $m \in\lambda(e)$.  Note that the choice of
$m\in \lambda(e)$ may not be unique. An edge can be an $m$- and $m'$-edge
at the same time.

To avoid trivial cases, we assume from here on that both the set $\X$ of
leaves and the set $M$ of colors is non-empty.
\begin{definition}\label{def:fitch-map}
  Let $\eXM$ be a map that assigns to every pair
    $(x,y)\in \irr{\X\times\X}$ a unique subset $M'\subseteq M$, where
    $M'=\emptyset$ may possible. Then, $\eps$ 
  is a \emph{Fitch map} if there is an edge-labeled tree $(T,\lambda)$ with
  leaf set $\X$ and edge labeling $\lambda: E(T)\to \PS{M}$ such that
  for every pair $(x,y)\in \irr{\X\times\X}$ holds
  \begin{equation*}
    m \in \eps(x,y) \iff \textnormal{ there is an }
    m\textnormal{-edge on the path from } \lca(x,y) \textnormal{ to } y.
  \end{equation*}
  In this case we say that $(T,\lambda)$ \emph{explains} $\eps$.
\end{definition}
\Cref{f:exmpl} provides an illustrative example of a Fitch map $\eps$ and
its corresponding tree $(T,\lambda)$.

A map $\eXM$ is called \emph{monochromatic} if $|M|=1$.  Monochromatic
Fitch maps are equivalent to the ``Fitch relations'' studied by
\citet{Geiss:18a}, and \citet{Hellmuth:19F}.

The map $\eps$ can also be interpreted as a set of $|M|$ not necessarily
disjoint binary relations (or equivalently graphs) on $\X$ defined by
the sets of pairs $\{(x,y) \in\irr{\X\times \X} \colon m \in \eps(x,y)\}$
(or equivalently arcs) for each fixed color $m\in M$. These relations
are disjoint if and only if $|\eps(x,y)|\le 1$ for every
$(x,y)\in\irr{\X\times \X}$, in which case we call $\eps$ a \emph{disjoint}
map.  Disjoint Fitch maps are equivalent to ``multi-colored Fitch graphs"
studied by \citet{Hellmuth:2019d}.

The Fitch maps defined here correspond to directed multi-graphs with the
restriction that there are no parallel arcs of the same color.  Note, we
may also allow parallel arcs with the same color $m$ provided that this
still means that there is an $m$-edge along the path from $\lca(x,y)$ to
$y$.  However, we must omit parallel edges with the same color $m$ whenever
the multiplicity $k$ of a parallel $m$-edge implies that at least $k$
$m$-edges must occur along the path from $\lca(x,y)$ to $y$, an issue that
may be part of future research.

\begin{figure}
\centering
\includegraphics[width=0.65\textwidth]{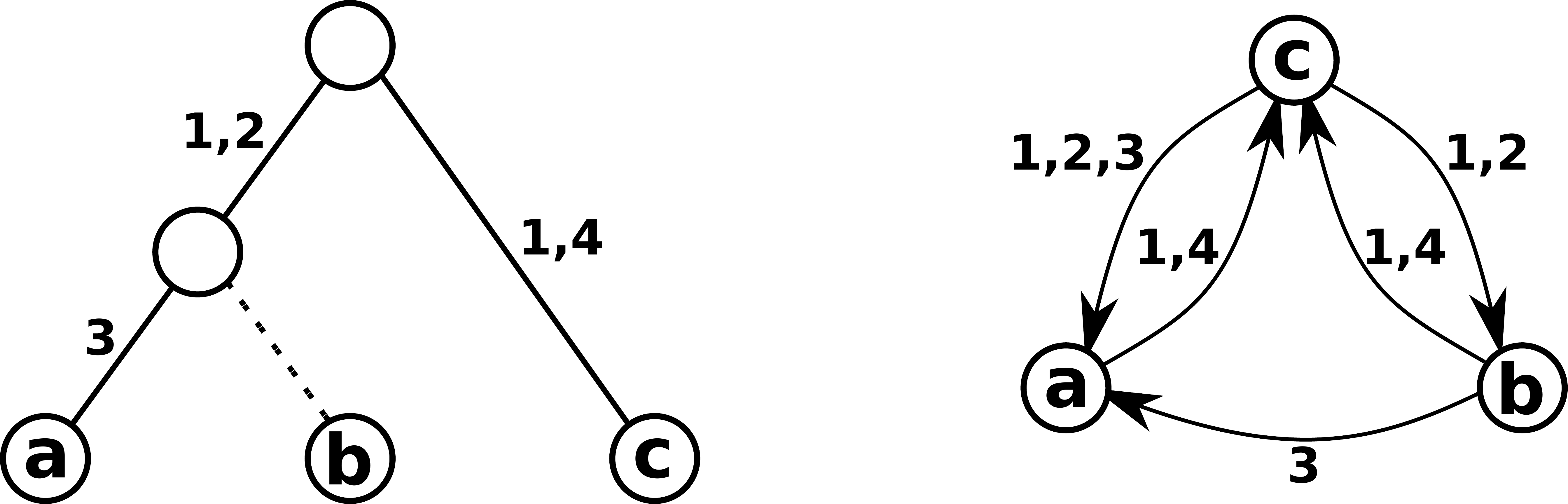}
\caption{The edge-labeled tree $(T,\lambda)$ with the leaf set
  $\mc{L}(T)=\{a,b,c\}=:\X$ on the left explains the displayed Fitch map
  $\eXM$ with the color set $M=\{1,2,3,4\}$ on the right.  It is easy to
  see that Fitch maps are not necessarily symmetric as e.g.\
  $\eps(a,b)\neq \eps(b,a)$.  Moreover, we can observe that
  $1\in \eps(a,c)$, $1\in \eps(c,b)$ but
  $1\notin\eps(a,b)=\emptyset$. Therefore, Fitch maps are not
  ``transitive'' in general.}
\label{f:exmpl}
\end{figure}

\subsection{Characterization in Terms of Neighborhoods}

We start by generalizing the approach developed by \citet{Hellmuth:19F} for
the monochromatic case.
\begin{definition} \label{def:neighbor}
  For a map $\eXM$, the set
  \begin{equation*}
    \NnotCol{m}[y] \coloneqq \{ x \in \X\setminus\{y\}\colon m \notin
    \eps(x,y)\}\cup \{y\}
  \end{equation*}
  is the \emph{(complementary) neighborhood $y\in\X$ for a color $m\in M$
    (w.r.t.\ $\eps$)}.\\
  We write $\Ns[\eps]\coloneqq \{ \NnotCol{m}[y] \colon y \in \X,\, m \in M \}$
  for the set of complementary neighborhoods of $\eps$.
\end{definition}
The set $\NnotCol{m}[y]\subseteq \X$ contains vertex $y$ and all
vertices $x \in \X\setminus\{y\}$ for which the color $m$ is \emph{not}
contained in $\eps(x,y)$. Informally speaking, if one thinks about a
di-graph that contains all arcs $(u,v)$ whenever $m \in\eps(u,v)$, then
$\NnotCol{m}[y]$ contains, in particular, all vertices $x$ that do
\emph{not} form an arc $(x,y)$. This fact justifies the name
``complementary'' neighborhood. Before we give an illustrative
example, we generalize the key conditions characterizing Fitch
relations in \cite{Hellmuth:19F} in terms of complementary
neighborhoods.
\begin{definition}
  A map $\eXM$ satisfies 
  \begin{itemize}[noitemsep]
  \item the \emph{hierarchy-like-condition (\HC)} if
    $\Ns[\eps]$ is hierarchy-like; and
  \item the \emph{inequality-condition (\IC)} if for every neighborhood
    $N\coloneqq\NnotCol{m}[y]\in \Ns[\eps]$ and every $y' \in N$, we have
    $|\NnotCol{m}[y']|\le|N|$.
  \end{itemize}
\end{definition}
The example in \cref{f:exmpl} gives some intuition for the definition of
the sets $\NnotCol{m}[y]$ and $\Ns[\eps]$: Here, we have
$\NnotCol{1}[b] = \{a,b\}$ and $\NnotCol{4}[b] = \{a,b,c\}$. In this
example, we obtain all clusters of size at least 2, and thus, all clusters
that are needed to recover the tree that explains the map $\eps$.  In fact,
$\Ns[\eps]$ is hierarchy-like.  However, even if $\Ns[\eps]$ is
hierarchy-like it may be the case that there is no tree that can explain
$\eps$, as we shall see below. As in \cite{Hellmuth:19F} the \IC will turn
out to be necessary as well.

The following proposition is crucial for the remaining part of this paper
as it provides a quite powerful characterization of neighborhoods and
edge-labeled trees that explain a Fitch map.
\begin{proposition}
  \label{p:clust-neighb}
  Let $(T,\lambda)$ be an edge-labeled tree explaining the Fitch map
  $\eXM$. Then, for every leaf $y \in \X$ and every color $m\in M$ there is
  a vertex $v\in V(T)$ such that the following two equivalent statements
  are satisfied:
  \begin{enumerate}
  \item a) There is no $m$-edge on the path from $v$ to $y$ and \\ b) the
    edge $(\parent(v),v)$ is an $m$-edge unless $v=\rootT{T}$.
  \item $\NnotCol{m}[y] = \Cl{T}(v)$.
  \end{enumerate}
\end{proposition}
\begin{proof}
  Let $\eXM$ be a Fitch map that is explained by $(T,\lambda)$.
  Furthermore, let $y \in \X$ be an arbitrary leaf and $m\in M$ be an
  arbitrary color.

  First, we show that there is a vertex $v \in V(T)$, which satisfies
  Statement~(1).  Let $v\in V(T)$ be the vertex that is an ancestor of $y$
  and is closest to the root $\rootT{T}$ such that there is no $m$-edge on
  the path from $v$ to $y$. Note that $v=y$ is possible.  By the choice of
  $v$, Statement~(1a) is trivially satisfied.  Now, assume that
  $v \ne \rootT{T}$.  This, together the the fact that $v$ is closest to
  the root among all ancestors of $y$ that satisfies that there is no
  $m$-edge on the path from $v$ to $y$, implies that $(\parent(v),v)$ is an
  $m$-edge.  Therefore, Statement~(1b) is also satisfied.

  Next, we show that Statement~(1) implies Statement~(2). For fixed $m$ and
  $y$, consider a vertex $v\in V(T)$ satisfying (1a) and (1b).

  First, we establish $v\preceq y$.  Since $y$ is a leaf, we either have
  $v\preceq y$ or $v$ and $y$ are not comparable.  Assume for contradiction
  that $v$ and $y$ are not comparable, and thus $v\ne \rootT{T}$.
  Therefore, Statement~(1b) implies that $(\parent(v),v)$ is an $m$-edge.
  However, $(\parent(v),v)$ lies on the path from $v$ to $y$; a
  contradiction to Statement~(1a). Thus, $v$ and $y$ must be comparable,
  and therefore $v\preceq y$.

  In order to see that $\Cl{T}(v) \subseteq \NnotCol{m}[y]$, we consider
  $x \in \Cl{T}(v)$, i.e.\ $v \preceq x$. This, together with $v \preceq y$
  implies that $v \preceq \lca(x,y) \preceq y$. Hence,
  $P\lc{T}(\lca(x,y),y)\subseteq P\lc{T}(v,y)$.  This, together with
  Statement~(1a), implies that there is no $m$-edge on the path from
  $\lca(x,y)$ to $y$. Since $(T,\lambda)$ explains $\eps$, we have
  $x \in \NnotCol{m}[y]$.

  Next, in order to see that $\NnotCol{m}[y] \subseteq \Cl{T}(v)$, we consider
  $x \in \NnotCol{m}[y]$. Note that $v\preceq y$ implies that
  $y\in \Cl{T}(v)$.  Therefore, if $x=y$, then $x=y\in \Cl{T}(v)$.  Now,
  assume that $x\ne y$. Hence, we have
  $x \in \NnotCol{m}[y]\setminus\{y\}$, and therefore $m \notin \eps(x,y)$.
  Thus, since $(T,\lambda)$ explains $\eps$, there is no $m$-edge on the
  path from $\lca(x,y)$ to $y$.  Note that if $\lca(x,y)\prec v \preceq y$,
  then $v\neq \rootT{T}$ and Statement~(1b) implies that $(\parent(v),v)$ is
  an $m$-edge that, in particular, is on the path from $\lca(x,y)$ to $y$;
  a contradiction.  Thus, $\lca(x,y)$ cannot be a strict ancestor of
  $v$. Moreover, $v \preceq y$ and $\lca(x,y)\preceq y$ imply that
  $\lca(x,y)$ and $v$ are comparable. The latter arguments together imply
  that $v \preceq \lca(x,y) \preceq x$. Hence, $x \in \Cl{T}(v)$.

  We proceed to show that Statement~(2) implies Statement~(1). Suppose that
  Statement~(2) is satisfied. If $v=y$, then Statement~(1a) is trivially
  satisfied. Now, assume that $v\ne y$. Choose an $x \in \X$ such that
  $\lca(x,y)=v$. Note that $v\ne y$ implies that $x$ and $y$ are
  distinct. This and $x \in \Cl{T}(v) = \NnotCol{m}[y]$ imply that
  $m \notin \eps(x,y)$. This and the fact that $(T,\lambda)$ explains
  $\eps$ imply that there is no $m$-edge on the path from $v=\lca(x,y)$
  to $y$. In summary, Statement~(1a) is satisfied.

  Now, assume that $v\ne \rootT{T}$. Hence, there is a parent $\parent(v)$
  of $v$.  Therefore, we can choose a vertex
  $x' \in \Cl{T}(\parent(v))\setminus\Cl{T}(v)$.  Hence,
  $\parent(v)=\lca(x',y)$. Since $x' \notin \Cl{T}(v)=\NnotCol{m}[y]$ and
  $x' \ne y$, we have $m \in \eps(x',y)$. This, together with the fact that
  $(T,\lambda)$ explains $\eps$, implies that there is an $m$-edge on the
  path from $\lca(x',y)=\parent(v)$ to $y$.  We have already shown that
  Statement~(1a) is satisfied, and thus, there is no $m$-edge on the path
  from $v$ to $y$. The latter two arguments imply that $ (\parent(v),v)$
  must be an $m$-edge. Therefore, Statement~(1b) is also true.
\end{proof}
  
\Cref{p:clust-neighb} has several simple but important consequences.
\begin{corollary} \label{cor:HC}
  Every Fitch map $\eps$ satisfies the hierarchy-like-condition (\HC), and
  $\Ns[\eps] \subseteq \Cls(T)$ for every tree $(T,\lambda)$ that explains
  $\eps$.
\end{corollary}
\begin{proof}
  Let $(T,\lambda)$ be an arbitrary tree that explains $\eps$.  By
  \cref{p:clust-neighb}, for every neighborhood
  $\NnotCol{m}[y] \in \Ns[\eps]$ there is always a vertex $v \in V(T)$ with
  $\NnotCol{m}[y]=\Cl{T}(v)\in\Cls(T)$.  Hence,
  $\Ns[\eps] \subseteq \Cls(T)$. By \Cref{lem:tree-iff-cluster}, the set
  $\Cls(T)$ forms a hierarchy.  Since every subset of the cluster-set
  $\Cls(T)$ of a phylogenetic tree $T$ is hierarchy-like, the map $\eps$
  satisfies \HC.
\end{proof}
Moreover, \cref{p:clust-neighb} can also be used to show
\begin{corollary}
  Every Fitch map $\eps$ satisfies the inequality-condition (\IC).
  \label{cor:IC}
\end{corollary}
\begin{proof}
  Let $\eXM$ be a Fitch map that is explained by $(T,\lambda)$.
  Furthermore, let $y \in \X$ be an arbitrary leaf and $m\in M$ be an
  arbitrary color; and let $y' \in N\coloneqq \NnotCol{m}[y]$. We need to
  show that $|\NnotCol{m}[y']|\le |N|$.
	
  By \cref{cor:HC}, we have $\Ns[\eps]\subseteq\Cls(T)$ and therefore, 
  $\NnotCol{m}[y'],N \in \Ns[\eps] \subseteq \Cls(T)$ are clusters. Hence,	
  there are vertices $v,v' \in V(T)$ such that $N=\Cl{T}(v)$ and
  $\NnotCol{m}[y']=\Cl{T}(v')$. Since $y' \in N=\Cl{T}(v)$ and
  by definition $y' \in \NnotCol{m}[y']=\Cl{T}(v')$, we have $v \preceq y'$
  and $v' \preceq y'$.  Hence, both $v$ and $v'$ are contained in the
  unique path from $y'$ to the root. Thus $v$ and $v'$ are comparable.

  Suppose that $v'$ is a strict ancestor of $v$, i.e.\ $v' \prec v$. Then,
  the edge $(\parent(v),v)$ lies on the path $P\lc{T}(v',y)$. Moreover,
  $\NnotCol{m}[y']=\Cl{T}(v')$ together with \Cref{p:clust-neighb}~(1a,~2)
  implies that there is no $m$-edge on the path $P\lc{T}(v',y)$.  The
  latter arguments together imply that the edge $(\parent(v),v)$ is not an
  $m$-edge. However, using \Cref{p:clust-neighb}~(1b, 2) for $N=\Cl{T}(v)$
  implies that $(\parent(v),v)$ is an $m$-edge; this is a contradiction.
  Therefore we must have $v \preceq v'$. This implies
  $\NnotCol{m}[y']=\Cl{T}(v')\subseteq\Cl{T}(v)=N$, and thus also the
  desired inequality $|\NnotCol{m}[y']|\le|N|$.
\end{proof}

Before we show that \HC and \IC are also sufficient conditions for Fitch
maps, we provide the following interesting result. Although this
  result does not have direct impact on the proofs for \HC and \IC, it
  provides interesting details about the relationship between Fitch maps
  $\eps$ and certain sets of triples that we can derive from $\eps$.

\begin{proposition}
  Let $\eXM$ be a Fitch map, and let
  \begin{equation*}
    \mc{R}(\eps) \coloneqq \bigcup_{N\in \Ns[\eps]} \{ ab|c : a,b\in N
    \text{ and } c\in \X\setminus N \}
  \end{equation*}
  be a triple set constructed of the neighborhoods of $\eps$.  Then, we
  have the following:
\begin{enumerate}%[noitemsep]
\item every edge-labeled tree $(T,\lambda)$ that explains $\eps$ displays
  all triples in $\mc{R}(\eps)$, i.e.,
  $\mc{R}(\eps)\subseteq \mathfrak{R}(T)$.
\item $\mc{R}(\eps)$ is closed, i.e.,
  $\overline{\mc{R}(\eps)} = \mc{R}(\eps)$.
\end{enumerate}
\end{proposition}
\begin{proof}
  Let $\eXM$ be a Fitch map, and let $(T,\lambda)$ be a tree that explains
  $\eps$. Then, assume that $ab|c \in \mc{R}(\eps)$. Thus, there is a
  neighborhood $N \in \Ns[\eps]$ with $a,b \in N$ and
  $c \in \X\setminus N$. By \cref{p:clust-neighb}, we conclude that
  $N = \Cl{T}(v)$ for some $v\in V(T)$.  Moreover, $c \in \X\setminus N$
  implies that there is a vertex $w\in V(T)$ with $w=\lca\lc{T}(a,b,c)$,
  and thus $\Cl{T}(v) \subsetneq \Cl{T}(w)$. The latter directly implies
  that $\lca\lc{T}(a,b,c) =w \prec\lc{T} v \preceq\lc{T} \lca\lc{T}(a,b)$, and
  thus $ab|c$ is displayed by $T$.

  We proceed with showing that $\mc{R}(\eps)$ is closed. First, we apply
  \Cref{cor:HC} to conclude that $\Ns[\eps]$ is hierarchy-like.  Thus,
  $\mc{N} \coloneqq \Ns[\eps] \cup \{\X\}\cup \{\{x\}: x \in \X\}$ is a
  hierarchy. By \Cref{lem:tree-iff-cluster}, there is a (unique) tree $T$
  such that $\Cls(T) = \mc{N}$. By construction, we conclude that
  $\mc{R}(\eps) = \mathfrak{R}(T)$, and thus
  $\overline{\mc{R}(\eps)} = \overline{\mathfrak{R}(T)}$.  Due to the
  definition of a closure, we have
  $\mc{R}(\eps) = \mathfrak{R}(T)\subseteq \overline{\mathfrak{R}(T)}$.
  However, since $T\in \langle \mathfrak{R}(T)\rangle$ and due to the
  definition of a closure, we obtain
  $\mathfrak{R}(T)=\overline{\mathfrak{R}(T)}$.  Therefore,
  $\mc{R}(\eps) = \mathfrak{R}(T) = \overline{\mathfrak{R}(T)} =
  \overline{\mc{R}(\eps)}$.  Hence, $\mc{R}(\eps)$ is closed.
\end{proof}

In order to show that \HC and \IC are also sufficient conditions for Fitch
maps, we define a particular edge-labeled tree $\To{\eps} = (T,\lambda)$
and proceed by proving that $\To{\eps}$ explains $\eps$.
\begin{definition} \label{def:tree-for-eps} Let $\eXM$ be a map that
  satisfies $\HC$.  The edge-labeled tree $\To{\eps} = (T,\lambda)$ on $\X$
  (with $M$), called \emph{\nameT,} has the cluster set
  \begin{equation}
    \Cls(T) = \Ns[\eps] \cup \big\{\X\big\} \cup
    \big\{ \{x\} \colon x \in \X \big\}
    \tag{a}
  \end{equation}
  and each edge $(\parent(v),v)$ of $T$ obtains the label
  \begin{equation}
    \lambda(\parent(v),v) \coloneqq
    \big\{ m \in M : \textnormal{ there is a } y\in \X
    \textnormal{ with }\Cl{T}(v) =\NnotCol{m}[y]\big\}.
    \tag{b}
  \end{equation}
\end{definition}

We first note that the \nameT $\To{\eps}=(T,\lambda)$ is
well-defined: If there is a map $\eXM$ that satisfies $\HC$, then $\Cls(T)$
is indeed a hierarchy. This, together with \Cref{lem:tree-iff-cluster},
implies that the phylogenetic tree $T$ on $\X$ is well-defined. The
edge-labeling $\lambda:E\to \PS{M}$ in \cref{def:tree-for-eps} requires
only the existence of vertices $y \in \X$ with $\Cl{T}(v)=\NnotCol{m}[y]$
for some $m\in M$, and thus, $\lambda$ is also well-defined.

Now, we are in the position to show that \HC and \IC are sufficient for
Fitch maps.
\begin{lemma} 
  \label{lem:sufficiency}
  Let $\eXM$ be a map that satisfies $\HC$ and $\IC$.  Then, $\eps$ is a
  Fitch map explained by $\To{\eps}$.
\end{lemma}
\begin{proof}
  Let $\eXM$ be a map that satisfies $\HC$ and $\IC$, and let
  $\To{\eps}=(T,\lambda)$ be the \nameT.  To show that $\eps$ is a
  Fitch map it suffices to show that $\To{\eps}$ explains $\eps$.  To this
  end, we will show that for every $(x,y)\in \irr{\X\times \X}$ we have the
  following:
  \begin{equation*}
    m \in \eps(x,y) \iff \textnormal{ there is an } m\textnormal{-edge on
      the path from } \lca(x,y) \textnormal{ to } y.
  \end{equation*}
  Let $(x,y) \in \irr{\X\times\X}$, and suppose that $m \in \eps(x,y)$.
  Hence, $x \notin \NnotCol{m}[y]$. By construction of $\To{\eps}$, the set
  $\NnotCol{m}[y]\in \Ns[\eps]\subseteq \Cls(T)$ is a cluster. Hence, there
  is a vertex $v\in V$ with $\NnotCol{m}[y]=\Cl{T}(v)$. Since
  $y\in \NnotCol{m}[y]=\Cl{T}(v)$, i.e.\ $v \preceq y$, and
  $\lca(x,y)\preceq y$, we conclude that $v$ and $\lca(x,y)$ are
  comparable. Moreover, since $x \notin \NnotCol{m}[y]=\Cl{T}(v)$, i.e.\
  $v \npreceq x$, and $\lca(x,y)\preceq x$, we can conclude that
  $v\npreceq \lca(x,y)$. The latter two arguments imply that
  $\lca(x,y) \prec v$; and therefore,
  $\lca(x,y) \preceq \parent(v)\prec v\preceq y$. Hence, the edge
  $(\parent(v),v)$ lies on the path from $\lca(x,y)$ to $y$.  Since
  $\Cl{T}(v)=\NnotCol{m}[y]$, $v\ne \rootT{T}$ and by the construction of
  $\To{\eps}$, we have $m \in \lambda(\parent(v),v)$, i.e.\
  $(\parent(v),v)$ is an $m$-edge. Hence, there is the $m$-edge
  $(\parent(v),v)$ that lies on the path from $\lca(x,y)$ to $y$.
	
  Conversely, suppose that there is an $m$-edge $(\parent(v),v)$ on the
  path from $\lca(x,y)$ to $y$ in $\To{\eps}$, and that
  $(x,y) \in \irr{\X\times\X}$. By construction of $\To{\eps}$, cf.\
  \cref{def:tree-for-eps}~(b), there is a leaf $y' \in \X$ with
  $\Cl{T}(v)=\NnotCol{m}[y']\eqqcolon N$.

  We continue to show that $\NnotCol{m}[y] \subseteq \Cl{T}(v)$. Since $v$
  lies on the path $P\lc{T}(\lca(x,y),y)$, we have $v \preceq y$, and thus
  $y \in \Cl{T}(v)$. By the construction of $\To{\eps}$, we have
  $\Ns[\eps]\subseteq\Cls(T)$; and therefore,
  $\NnotCol{m}[y],\Cl{T}(v) \in\Cls(T)$ are clusters. This, together with
  $y \in \Cl{T}(v)\cap \NnotCol{m}[y] \ne \emptyset$, implies either
  $\NnotCol{m}[y]\subseteq\Cl{T}(v)$ or
  $\Cl{T}(v)\subsetneq\NnotCol{m}[y]$. Moreover, since $\eps$ satisfies the
  \IC and $y \in N = \Cl{T}(v)$ it must hold that
  $|\NnotCol{m}[y]|\le|N|=|\Cl{T}(v)|$. The latter two arguments
  immediately imply $\NnotCol{m}[y]\subseteq \Cl{T}(v)$. Furthermore, since
  $(\parent(v),v)$ lies on the path from $\lca(x,y)$ to $y$, we can
  conclude that $x \notin \Cl{T}(v)$. This and
  $\NnotCol{m}[y]\subseteq \Cl{T}(v)$ imply that $x \notin
  \NnotCol{m}[y]$. Hence, by definition of $\NnotCol{m}[y]$ we must have
  $m \in \eps(x,y)$.
	
  To summarize, for any pair $(x,y) \in \irr{\X\times\X}$ we have
  $m \in \eps(x,y)$ if and only if there is an $m$-edge on the path from
  $\lca(x,y)$ to $y$ in $\To{\eps}$. Therefore, $\To{\eps}$ explains
  $\eps$; and thus, $\eps$ is a Fitch map.
\end{proof}

\Cref{cor:HC,cor:IC}, together with \Cref{lem:sufficiency}, imply
\begin{theorem} \label{thm:charact-HCIC} A map $\eXM$ is a Fitch map if and
  only if
  \begin{enumerate}[noitemsep]
  \item $\Ns[\eps]$ is hierarchy-like (\HC); and
  \item for every neighborhood $N\coloneqq\NnotCol{m}[y] \in \Ns[\eps]$ and
    every leaf $y' \in N$, we have $|\NnotCol{m}[y']|\le|N|$ (\IC).
  \end{enumerate}
\end{theorem}
Note that \Cref{thm:charact-HCIC} and \cite[Thm.\ 4]{Hellmuth:19F} are
equivalent in case that $\eps$ is a monochromatic Fitch map.  In
  particular, for a Fitch map $\eps$, \Cref{thm:charact-HCIC} implies that
  $\NnotCol{m}[y'] \subseteq N$ for every neighborhood
  $N\coloneqq\NnotCol{m}[y] \in \Ns[\eps]$ and every leaf $y' \in N$, since
  $\Ns[\eps]$ must be hierarchy-like.

\subsection{Characterization in Terms of Forbidden Submaps}

Monochromatic Fitch maps $\eXM$ with $|M|=1$ are characterized by a small
set of forbidden subgraphs \cite{Geiss:18a,Hellmuth:18b}. In what follows,
we show that also non-monochromatic Fitch maps have a forbidden submap
characterization as defined as follows:

\begin{definition}
Let $\eXM$ and $\eps':\irr{\X'\times\X'}\to \PS{M'}$ be two maps. 
Then, the map $\eps'$ is a \emph{submap of $\eps$} if $\X' \subseteq \X$, and
$\eps'(x,y)\subseteq \eps(x,y)$ for every $(x,y)\in \irr{\X' \times \X'}$. 
In addition, a submap $\eps'$ is an \emph{induced} submap of $\eps$ if
$\eps'(x,y)= \eps(x,y)$ for every $(x,y)\in \irr{\X' \times \X'}$.
\end{definition}

For the characterization in terms of forbidden submaps, we first provide
the next lemma, which is illustrated in \cref{f:forb-gen-rel}.
\begin{figure}[t]
  \centering \includegraphics[width=0.65\textwidth]{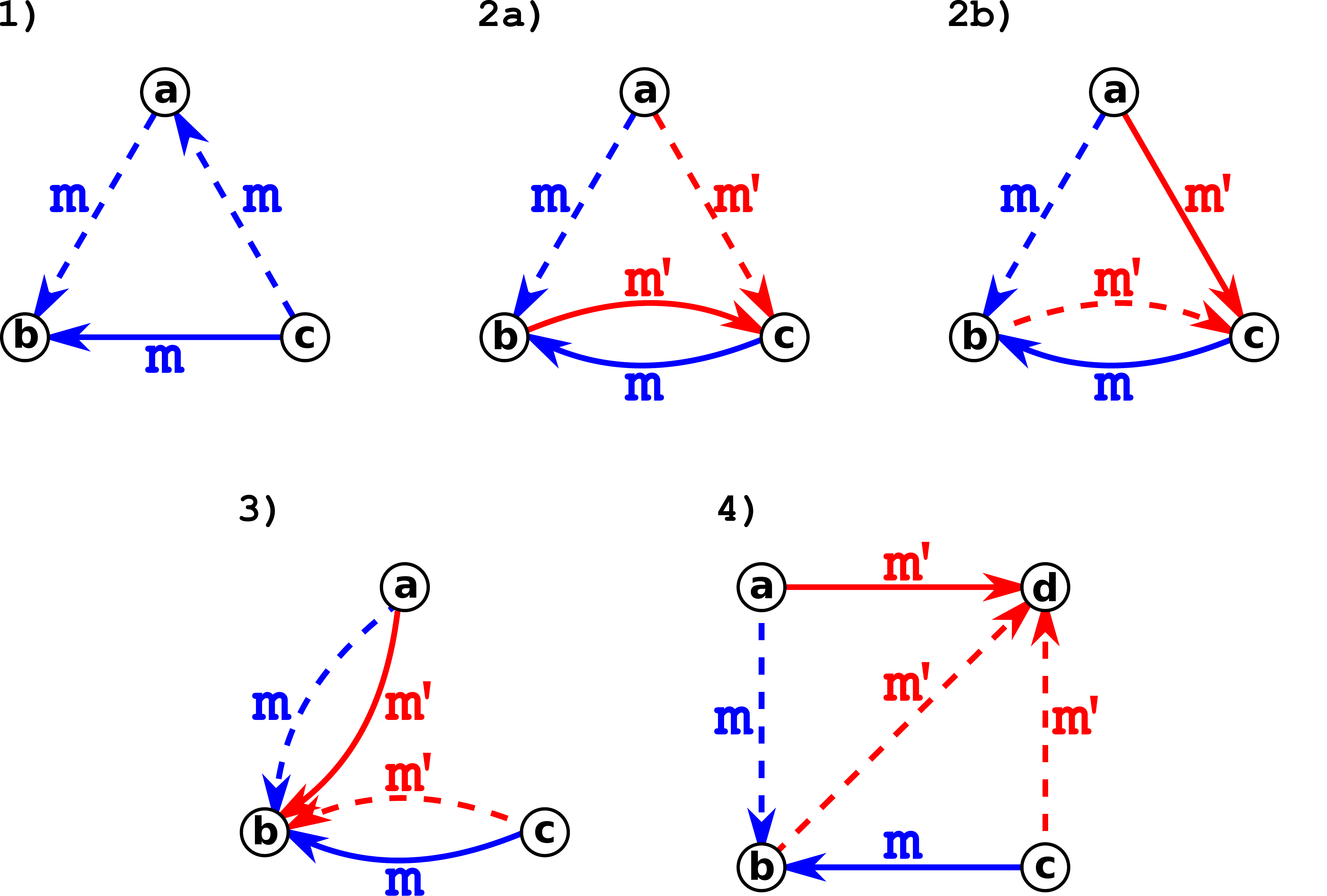}
  \caption{%
    Based on \Cref{lem:non-fitch}, there are five forbidden (not
    necessarily induced) submaps on one or two colors.  The colored
    graph-representation of these maps are shown here.  Solid edges
    indicate that the particular color must occur, while dashed edges
    indicate that the particular color must not occur.  Note, $m,m'$ are
    not necessarily distinct except for the Cases (3) and (4).}
  \label{f:forb-gen-rel}
\end{figure}

\begin{lemma}
  A map $\eXM$ is \emph{not} a Fitch map if and only if there are (not
  necessarily distinct) colors $m,m' \in M$ and
  \begin{itemize}
  \item there is a subset $\{a,b,c\} \in \binom{\X}{3}$ with
    $m \in \eps(c,b)$ and
    $m \notin \eps(a,b)$ that satisfies one of the following conditions
    \begin{enumerate}[noitemsep]
    \item $m \notin \eps(c,a)$, or
    \item a) $m' \notin \eps(a,c)$ and $m' \in \eps(b,c)$, or \\
      b) $m' \in \eps(a,c)$ and $m' \notin \eps(b,c)$, or	
    \item $m\ne m'$, $m' \notin \eps(c,b)$ and $m' \in \eps(a,b)$; or
    \end{enumerate}
  \item there is a subset $\{a,b,c,d\} \in \binom{\X}{4}$
    with $m \in \eps(c,b)$ and
    $m \notin \eps(a,b)$ that satisfies 
    \begin{enumerate}[noitemsep]
    \item[4.] $m\ne m'$, $m' \notin \eps(b,d)\cup \eps(c,d)$ and
      $m' \in \eps(a,d)$.
    \end{enumerate}
  \end{itemize}
  \label{lem:non-fitch}
\end{lemma}
\begin{proof}
  Let $\eXM$ be an arbitrary map. First, suppose that $\eps$ is not a Fitch
  map.  Thus, \Cref{thm:charact-HCIC} implies that $\eps$ does not
  satisfy \HC or \IC.
	
  First, suppose that $\eps$ does not satisfy \IC. Hence, there is a
  neighborhood $\NnotCol{m}[b] \in \Ns[\eps]$ and a vertex
  $a \in \NnotCol{m}[b]$ 
  such that $|\NnotCol{m}[a]|>|\NnotCol{m}[b]|$. Hence, $a \ne b$ and there
  is a vertex $c \in \NnotCol{m}[a]\setminus\NnotCol{m}[b]$. Since
  $a,b\in \NnotCol{m}[b]$, we have $c \ne a$ and $c \ne b$, and hence
  $\{a,b,c\} \in { \X \choose 3}$. Since $c \notin \NnotCol{m}[b]$, it must
  hold that $m \in \eps(c,b)$. Since $a \in \NnotCol{m}[b]$, it must hold
  that $m \notin \eps(a,b)$. Since $c \in \NnotCol{m}[a]$, it must hold
  that $m \notin \eps(c,a)$. The last three observations imply that
  Condition (1) is satisfied.
	
  Now, assume that $\eps$ does not satisfy \HC, and thus that $\Ns[\eps]$
  is not hierarchy-like. Hence, there are two neighborhoods
  $\NnotCol{m}[y],\NnotCol{m'}[y'] \in \Ns[\eps]$ such that
  $\NnotCol{m}[y]\cap \NnotCol{m'}[y'] \notin
  \{\emptyset,\NnotCol{m}[y],\NnotCol{m'}[y']\}$. This, together with the fact that
  $y \in \NnotCol{m}[y]$ and $y' \in \NnotCol{m'}[y']$, implies that there
  are three mutually exclusive cases that need to be examined:
  \begin{itemize}[noitemsep]
  \item[{(A)}] neither of $y$ and $y'$ is contained in
    $\NnotCol{m}[y]\cap \NnotCol{m'}[y']$,
  \item[{(B)}] exactly one element of $\{y,y'\}$ is contained in
    $\NnotCol{m}[y]\cap \NnotCol{m'}[y']$,
  \item[{(C)}] both $y$ and $y'$ are contained in
    $\NnotCol{m}[y]\cap \NnotCol{m'}[y']$.
  \end{itemize}

  First, consider Case (A). This case is equivalent to
  $y \in \NnotCol{m}[y]\setminus\NnotCol{m'}[y']$ and
  $y' \in \NnotCol{m'}[y']\setminus\NnotCol{m}[y]$. Since
  $\NnotCol{m}[y]\cap\NnotCol{m'}[y']\ne \emptyset$, there is a vertex
  $a\in \NnotCol{m}[y]\cap\NnotCol{m'}[y']$ with $y,y' \ne a$. Thus, $a,y$
  and $y'$ are pairwise distinct. Since $y \notin \NnotCol{m'}[y']$, we
  have $m' \in \eps(y,y')$, and since $y' \notin \NnotCol{m}[y]$, we have
  $m \in \eps(y',y)$.  Moreover, since
  $a \in \NnotCol{m}[y]\cap \NnotCol{m'}[y']$, we have $m \notin\eps(a,y)$
  and $m' \notin \eps(a,y')$. Now, put $b \coloneqq y$ and
  $c \coloneqq y'$. Then, we have found a subset
  $\{a,b=y,c=y'\} \in {\X \choose 3 }$ such that $m \in \eps(c,b)$,
  $m \notin\eps(a,b)$, $m' \notin \eps(a,c)$ and $m' \in \eps(b,c)$. Hence,
  Condition (2a) is satisfied.
	
  Now, consider Case (B). We can assume w.l.o.g.\ that
  $y \in \NnotCol{m}[y]\cap \NnotCol{m'}[y']$ and
  $y' \in \NnotCol{m'}[y']\setminus\NnotCol{m}[y]$. Since
  $\NnotCol{m}[y]\setminus\NnotCol{m'}[y']\ne \emptyset$, there is a vertex
  $a\in \NnotCol{m}[y]\setminus\NnotCol{m'}[y']$ with $y,y'\ne a$. Thus,
  $a,y$ and $y'$ are pairwise distinct. Since $y \in \NnotCol{m'}[y']$, we
  have $m' \notin \eps(y,y')$, and since $y' \notin \NnotCol{m}[y]$, we
  have $m \in \eps(y',y)$. Moreover, since
  $a \in \NnotCol{m}[y]\setminus\NnotCol{m'}[y']$, we have
  $m \notin\eps(a,y)$ and $m' \in \eps(a,y')$. Now, put $b \coloneqq y$ and
  $c \coloneqq y'$. Then, we have found a subset
  $\{a,b=y,c=y'\} \in {\X \choose 3 }$ such that $m \in \eps(c,b)$,
  $m \notin\eps(a,b)$, $m' \in \eps(a,c)$ and $m' \notin \eps(b,c)$. Hence,
  Condition (2b) is satisfied.
	
  Next, consider Case (C). Here we consider the two subcases (i) $y=y'$ and
  (ii) $y\neq y'$. Let us start with Subcase (C.i) and suppose that $y=y'$.
  Since
  $\NnotCol{m}[y]\cap \NnotCol{m'}[y] \notin
  \{\emptyset,\NnotCol{m}[y],\NnotCol{m'}[y]\}$, we can directly conclude
  that $m\neq m'$. Moreover, since
  $\NnotCol{m}[y]\setminus\NnotCol{m'}[y]\ne \emptyset$ and
  $\NnotCol{m'}[y]\setminus\NnotCol{m}[y]\ne \emptyset$, there are two
  distinct vertices $a \in \NnotCol{m}[y]\setminus\NnotCol{m'}[y]$ and
  $c \in \NnotCol{m'}[y]\setminus\NnotCol{m}[y]$. Hence, $a,c$ and $y$ are
  pairwise distinct. Since $a \in \NnotCol{m}[y]\setminus\NnotCol{m'}[y]$,
  we have $m \notin \eps(a,y)$ and $m' \in \eps(a,y)$.  Moreover, since
  $c \in \NnotCol{m'}[y]\setminus\NnotCol{m}[y]$, we have
  $m' \notin\eps(c,y)$ and $m \in \eps(c,y)$. Now, put $b \coloneqq y$.
  Then, we have found a subset $\{a,b=y,c\} \in {\X \choose 3 }$ such that
  $m \in \eps(c,b)$, $m \notin\eps(a,b)$, $m' \notin \eps(c,b)$ and
  $m' \in \eps(a,b)$. This, together with $m\ne m'$, implies that Condition
  (3) is satisfied.
	
  Finally, consider Subcase (C.ii) and suppose that $y\neq y'$.  Since
  $\NnotCol{m}[y]\setminus\NnotCol{m'}[y']\ne \emptyset$ and
  $\NnotCol{m'}[y']\setminus\NnotCol{m}[y]\ne \emptyset$, there are two
  distinct vertices $x \in \NnotCol{m}[y]\setminus\NnotCol{m'}[y']$ and
  $x' \in \NnotCol{m'}[y']\setminus\NnotCol{m}[y]$. Hence, $x,x',y$ and
  $y'$ are pairwise distinct, since
  $y,y' \in \NnotCol{m}[y]\cap \NnotCol{m'}[y']$ are distinct. Since
  $x \in \NnotCol{m}[y]\setminus\NnotCol{m'}[y']$, we have
  $m \notin \eps(x,y)$ and $m' \in \eps(x,y')$. Since
  $y \in \NnotCol{m'}[y']$, we have $m' \notin \eps(y,y')$. Moreover, since
  $x' \in \NnotCol{m'}[y']\setminus\NnotCol{m}[y]$, we have
  $m' \notin\eps(x',y')$ and $m \in \eps(x',y)$.
	
  Now, assume that $m=m'$. Then, put $a\coloneqq y$, $b\coloneqq y'$ and
  $c\coloneqq x$. Thus, we have found a subset
  $\{a=y,b=y',c=x\} \in {\X \choose 3 }$ such that $m=m' \in \eps(c,b)$,
  $m=m' \notin\eps(a,b)$, $m\notin \eps(c,a)$. Hence, Condition (1) is
  satisfied.
	
  Next, assume that $m\ne m'$. Then, put $a\coloneqq x$, $b \coloneqq y$,
  $c\coloneqq x'$ and $d \coloneqq y'$. Then, we have found a subset
  $\{a=x,b=y,c=x',d=y'\} \in {\X \choose 4 }$ such that $m \in \eps(c,b)$,
  $m \notin\eps(a,b)$, $m' \notin\eps(b,d)\cup\eps(c,d)$ and
  $m' \in \eps(a,d)$. This, together with $m\ne m'$, implies that Condition
  (4) is satisfied.
	
  In summary, if $\eps$ is not a Fitch map, then at least one of the
  Conditions (1), (2), (3) and (4) must be satisfied.
  \smallskip \\
  Conversely, suppose that $\eXM$ satisfies at least one of the Conditions
  (1), (2) (3) or (4).  First, assume that Condition (1) holds. Then,
  $a\in \NnotCol{m}[a]\cap\NnotCol{m}[b]$ and
  $c \in \NnotCol{m}[a]\setminus\NnotCol{m}[b]$.  If
  $\NnotCol{m}[b]\nsubseteq \NnotCol{m}[a]$, then $\Ns[\eps]$ is not
  hierarchy-like; and therefore, \Cref{thm:charact-HCIC} implies that
  $\eps$ is not a Fitch map.  Since
  $c \in \NnotCol{m}[a]\setminus\NnotCol{m}[b]$ implies that
  $\NnotCol{m}[a]\ne\NnotCol{m}[b]$, we can now assume that
  $\NnotCol{m}[b] \subsetneq \NnotCol{m}[a]$.  Hence, there is a
  neighborhood $N\coloneqq \NnotCol{m}[b]$ and a vertex $a \in N$ such that
  $|\NnotCol{m}[a]|>|N|$; and therefore, $\eps$ does not satisfy the
  inequality-condition (\IC). \Cref{thm:charact-HCIC} implies that $\eps$
  is not a Fitch map.  Hence, either way, if $\eps$ satisfies Condition
  (1), then $\eps$ is not a Fitch map.
	
  Now, if Condition (2a) is satisfied, then
  $a \in \NnotCol{m}[b]\cap \NnotCol{m'}[c]$,
  $b \in \NnotCol{m}[b]\setminus\NnotCol{m'}[c]$ and
  $c \in \NnotCol{m'}[c]\setminus\NnotCol{m}[b]$. If Condition (2b) is
  satisfied, then $a \in \NnotCol{m}[b]\setminus \NnotCol{m'}[c]$,
  $b \in \NnotCol{m}[b]\cap\NnotCol{m'}[c]$ and
  $c \in \NnotCol{m'}[c]\setminus\NnotCol{m}[b]$. Next, if Condition (3)
  is satisfied, then $a \in \NnotCol{m}[b]\setminus \NnotCol{m'}[b]$,
  $b \in \NnotCol{m}[b]\cap\NnotCol{m'}[b]$ and
  $c \in \NnotCol{m'}[b]\setminus\NnotCol{m}[b]$. Moreover, if Condition
  (4) is satisfied, then $a \in \NnotCol{m}[b]\setminus \NnotCol{m'}[d]$,
  $b \in \NnotCol{m}[b]\cap\NnotCol{m'}[d]$ and
  $c \in \NnotCol{m'}[d]\setminus\NnotCol{m}[b]$. It is easy to see that in
  neither case the set $\Ns[\eps]$ is hierarchy-like, and
  \cref{thm:charact-HCIC} implies that $\eps$ is not a Fitch map.
	
  In summary, if one of the Conditions (1), (2), (3) or (4) is satisfied, then
  $\eps$ is not a Fitch map.
\end{proof}

Application of simple Boolean conversion on \Cref{lem:non-fitch} implies
\begin{theorem}\label{thm:forbid-sub}
  A map $\eXM$ is a Fitch map if and only if for every (not necessarily
  distinct) colors $m,m' \in M$, and
  \begin{itemize}
  \item for every subset $\{a,b,c\} \in \binom{\X}{3}$ with
    $m \in \eps(c,b)$ and $m \notin \eps(a,b)$ we have
    \begin{enumerate}[noitemsep]
    \item $m \in \eps(c,a)$, and
    \item $m' \in \eps(a,c)$ if and only if $m' \in \eps(b,c)$, and
    \item if $m\ne m'$ and $m' \notin \eps(c,b)$, then
      $m' \notin \eps(a,b)$; and
    \end{enumerate}
  \item for every subset $\{a,b,c,d\} \in \binom{\X}{4}$ with
    $m \in \eps(c,b)$ and $m \notin \eps(a,b)$ we have
    \begin{enumerate}[noitemsep]
    \item[4.] if $m\neq m' $ and $m' \notin \eps(b,d)\cup \eps(c,d)$, then
      $m' \notin \eps(a,d)$.
    \end{enumerate}
  \end{itemize}
\end{theorem}

Note that \Cref{thm:forbid-sub} and \cite[Thm.\ 5]{Hellmuth:19F} are
equivalent in case that $\eps$ is a monochromatic Fitch map.  Moreover, the
characterization in \Cref{thm:forbid-sub} directly implies the following
result that shows that the recognition of Fitch maps reduces to the
recognition of Fitch maps on less than five vertices and on one or two
colors only.
\begin{corollary}
  Let $X'\subseteq X$, $\eXM$ be a map, and
  $\eps' \colon \irr{\X' \times \X'} \to \PS{M'}$ be the submap of $\eps$
  defined by $\eps'(x,y)\coloneqq\eps(x,y)\cap M'$ for every
  $(x,y) \in \irr{\X' \times \X'}$.  Then, the following statements are
  equivalent:
  \begin{enumerate}
  \item $\eps$ is a Fitch map.
  \item If $|M|\ge 2$ and $|\X|\ge 4$, then for every
    $\X' \in {\X\choose 4}$ and for every $M'\in {M\choose 2}$ the map
    $\eps'$ is a Fitch map.
  \item If $|M|\ge 2$ and $|\X|\le 3$, then for every $M'\in {M\choose 2}$
    the map $\eps'$ is a Fitch map, where $\X'\coloneqq \X$.
  \item If $|M|=1$, then for every $\X' \in {\X\choose 3}$ the map $\eps'$
    is a Fitch map, where $M'\coloneqq M$.
  \end{enumerate}
\end{corollary}
In addition, we obtain the following
\begin{corollary}\label{cor:induced-Fitch}
  Every induced submap of a Fitch map is a Fitch map.
\end{corollary}
\begin{proof}
  Let $\eXM$ be a Fitch map, and let $\eps':\irr{\X'\times \X'}\to \PS{M'}$
  be an induced submap of $\eps$.  Assume for contraposition that $\eps'$
  is not a Fitch map.  Then, \cref{thm:forbid-sub} implies that $\eps'$
  contains a forbidden submap.  Since $\eps'$ is an \emph{induced} submap
  of $\eps$, this forbidden submap is also part of $\eps$.  Thus,
  \cref{thm:forbid-sub} implies that $\eps$ is not a Fitch map.
\end{proof}

\section{Uniqueness of the Least-Resolved Tree}
\label{sec:uniq}

In the following we are interested in so-called least-resolved trees that
explain a given Fitch map $\eps$.  Roughly speaking, an edge-labeled tree $(T,\lambda)$
that explains $\eps$ is least-resolved if it does not contain
``unnecessary'' colors along its edges and one cannot ``contract'' edges
without destroying the property that the resulting tree still explains
$\eps$.  To make this notion more precise, we first need a couple of
definitions.
 
\begin{definition}
Let $(T,\lambda)$ and $(T',\lambda')$ be two edge-labeled trees on $\X$
with $M$. Then, $(T',\lambda')$ is a \emph{coarse-graining} of
$(T,\lambda)$, denoted by $(T',\lambda')\le(T,\lambda)$, if
\begin{itemize}[noitemsep]
\item $\Cls(T')\subseteq\Cls(T)$ and 
\item for each $v' \in V(T')\setminus\{\rootT{T'}\}$
and for each $v \in V(T)\setminus\{\rootT{T}\}$ with
$\Cl{T}(v)=\Cl{T'}(v')$ we have
$\lambda'(\parent\lc{T'}(v'),v')\subseteq\lambda(\parent\lc{T}(v),v)$. 
\end{itemize}
Moreover, a coarse-graining $(T',\lambda')$ of $(T,\lambda)$ is a
\emph{strict coarse-graining} of $(T,\lambda)$ whenever
\begin{itemize}[noitemsep]
\item $\Cls(T')\subsetneq\Cls(T)$  or 
\item for some $v' \in V(T')\setminus\{\rootT{T'}\}$ and some
  $v \in V(T)\setminus\{\rootT{T}\}$ with $\Cl{T}(v)=\Cl{T'}(v')$ we
  have
  $\lambda'(\parent\lc{T'}(v'),v')\subsetneq\lambda(\parent\lc{T}(v),v)$.
\end{itemize}
\end{definition}
In particular, we say $(T,\lambda)$ and $(T',\lambda')$ are \emph{isomorphic},
denoted by $(T,\lambda)\cong(T',\lambda')$, if they are coarse-grainings of
each other.

\begin{definition}
  Let $\eXM$ be a Fitch map that is explained by some edge-labeled tree
  $(\Ts,\ls)$.  Then, we say $(\Ts,\ls)$ is \emph{least-resolved w.r.t.\
    $\eps$} if there is no strict coarse-graining $(T',\lambda')$ of
  $(\Ts,\ls)$ that explains $\eps$.
\end{definition}
\Cref{f:least-resolved} provides an example of coarse-graining and
least-resolved edge-labeled trees.

\begin{figure}[t]
  \centering \includegraphics[width=0.9\textwidth]{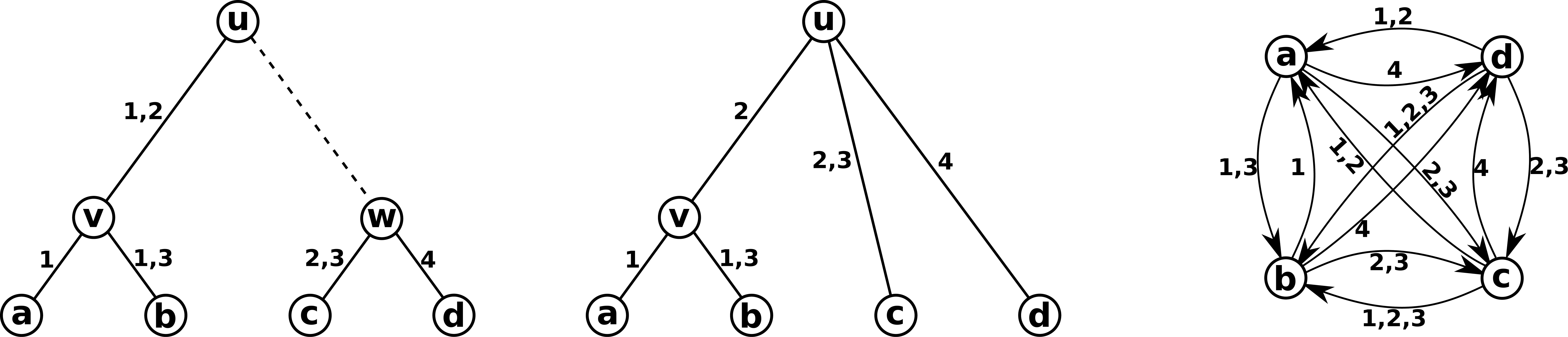}
  \caption{Both trees explain the Fitch map $\eps$ on the right. However,
    only the middle tree $(T',\lambda')$ is least-resolved w.r.t.\ $\eps$.
    In particular, \cref{thm:uniqueLRT} implies that $(T',\lambda')$ is
    isomorphic to the \nameT $\To{\eps}$.  Moreover, the tree
    $(T',\lambda')$ is a (strict) coarse-graining of the left tree
    $(T,\lambda)$ since $\Cls(T') = \Cls(T)\setminus \{\Cl{T}(w)\}$ as well
    as $\lambda'(u,v) = \lambda(u,v)\setminus\{1\}$.}
\label{f:least-resolved}
\end{figure}

\begin{proposition}
  \label{prop:coarse-graining}
  Let $\eXM$ be a Fitch map, let $(T,\lambda)$ be an arbitrary edge-labeled
  tree that explains $\eps$, and let $\To{\eps}$ be the \nameT.
  Then, $\To{\eps}$ is a coarse-graining of $(T,\lambda)$, i.e.\
  $\To{\eps}\le(T,\lambda)$.
\end{proposition}
\begin{proof}
  Let $\eXM$ be a Fitch map. Let $(T,\lambda)$ be an edge-labeled tree that
  explains $\eps$, and let $\To{\eps}=(\wh{T},\wh{\lambda})$ be the
  \nameT.
  
  First, \Cref{thm:charact-HCIC} implies that $\eps$ satisfies \HC and \IC.
  Moreover, \Cref{lem:sufficiency} implies that
  $\To{\eps}=(\wh{T},\wh{\lambda})$ explains $\eps$.  By construction of
  $\To{\eps}$, we have
  $\Cls(\wh{T}) = \Ns[\eps]\cup\{\X\}\cup\{ \{x\} : x\in \X \}$.  Since
  $(T,\lambda)$ explains $\eps$, we can apply \Cref{p:clust-neighb} to
  conclude that $\Ns[\eps] \subseteq \Cls(T)$. Since $\{x\}\in \Cls(T)$ for
  all $x\in \X$ and $\X \in \Cls(T)$, we immediately obtain
  $\Cls(\wh{T}) = \Ns[\eps]\cup\{\X\}\cup\{ \{x\} : x\in \X \} \subseteq
  \Cls(T)$.
	
  It remains to show that for each
  $\hat{v} \in V(\wh{T})\setminus\{\rootT{\wh{T}}\}$ and for each
  $v \in V(T)\setminus\{\rootT{T}\}$ with $\Cl{\wh{T}}(\hat{v})=\Cl{T}(v)$
  we have
  $\wh{\lambda}(\parent\lc{\wh{T}}(\hat{v}),\hat{v})\subseteq
  \lambda(\parent\lc{T}(v),v)$.  Thus, let
  $\hat{v} \in V(\wh{T})\setminus\{\rootT{\wh{T}}\}$ be an arbitrary
  vertex.  Since $\Cls(\wh{T}) \subseteq \Cls(T)$ there is a vertex
  $v \in V(T)\setminus\{\rootT{T}\}$ such that
  $\Cl{\wh{T}}(\hat{v})=\Cl{T}(v)$.  Let
  $m \in \wh{\lambda}(\parent\lc{\wh{T}}(\hat{v}),\hat{v})$ be an arbitrary
  color.  By construction of $\To{\eps}$, cf.\ \cref{def:tree-for-eps}~(b),
  there is a leaf $y \in \X$ with
  $\Cl{\wh{T}}(\hat{v})=\NnotCol{m}[y]$. Hence, we have
  $y\in \NnotCol{m}[y]=\Cl{\wh{T}}(\hat{v})=\Cl{T}(v)$.

  Next, assume that $v \in \X$ is a leaf, i.e.\ $\{v\}=\Cl{T}(v)$, and
  therefore $y=v$. Since $T$ is a phylogenetic tree, there is a leaf
  $z \in \X$ such that $\lca\lc{T}(v,z)=\parent\lc{T}(v)$. Moreover,
  $z \notin \{v\}=\NnotCol{m}[y]$ implies $m \in \eps(z,y)$. Since
  $(T,\lambda)$ explains $\eps$, we conclude that there is an $m$-edge on
  the path from $\lca\lc{T}(y,z) = \lca\lc{T}(v,z)=\parent\lc{T}(v)$ to
  $v$. Hence, $m \in \lambda(\parent\lc{T}(v),v)$.
	
  Now, assume that $v \notin \X$ is not a leaf. Recap,
  $y\in \NnotCol{m}[y]=\Cl{\wh{T}}(\hat{v})=\Cl{T}(v)$ and
  $v\neq \rootT{T}$.  This, together with the fact that $T$ is a
  phylogenetic tree, implies that there are two leaves $x,z \in \X$ such
  that $\lca\lc{T}(x,y)=v$ and $\lca\lc{T}(y,z)=\parent\lc{T}(v)$. Since
  $\lca\lc{T}(x,y)=v$ we have $x \in \Cl{T}(v)=\NnotCol{m}[y]$. Moreover,
  $\lca\lc{T}(y,z)=\parent\lc{T}(v)$ implies
  $z \notin \Cl{T}(v)=\NnotCol{m}[y]$. Therefore, $x$, $y$ and $z$ are
  pairwise distinct. The latter arguments imply that $m \in \eps(z,y)$ and
  $m\notin\eps(x,y)$. This, together with the fact that $(T,\lambda)$
  explains $\eps$, implies that there is an $m$-edge on the path from
  $\lca\lc{T}(y,z)=\parent\lc{T}(v)$ to $y$ and there is no $m$-edge on the
  path from $\lca\lc{T}(x,y)=v$ to $y$. Therefore, we have
  $m \in \lambda(\parent\lc{T}(v),v)$.
	
  Thus, $m \in \lambda(\parent\lc{T}(v),v)$ independent of whether
  $v \in \X$ or $v \notin \X$. Therefore, we have
  $\wh{\lambda}(\parent\lc{\wh{T}}(\hat{v}),\hat{v})\subseteq
  \lambda(\parent\lc{T}(v),v)$. This, together with
  $\Cls(\wh{T}) \subseteq \Cls(T)$, implies that $\To{\eps}$ is a
  coarse-graining of $(T,\lambda)$.
\end{proof}

\begin{theorem}
  Let $\eXM$ be a Fitch map and let $\To{\eps}=(\wh{T},\wh{\lambda})$ be
  the \nameT.   Then, $\To{\eps}$ is the unique (up to isomorphism)
  least-resolved tree that explains $\eps$.  In particular, $\wh{T}$ has
  the minimum number of vertices, and the sum
  $\sum_{e\in E(\wh{T})}|\wh{\lambda}(e)|$ is minimum among all
  edge-labeled trees that explain $\eps$.
  \label{thm:uniqueLRT}
\end{theorem}
\begin{proof}
  Let $\eXM$ be a Fitch map, and let $(\Ts,\ls)$ be a least-resolved
  edge-labeled tree w.r.t.\ $\eps$. Moreover, let
  $\To{\eps}=(\wh{T},\wh{\lambda})$ be the \nameT.
  % the edge-labeled tree as in \cref{def:tree-for-eps}.
	
  By \cref{prop:coarse-graining}, $\To{\eps} = (\wh{T},\wh{\lambda})$ is a
  coarse-graining of $(\Ts,\ls)$.  Hence, $\To{\eps}$ must be
  least-resolved. This, together with the fact that $(\Ts,\ls)$ is a
  least-resolved tree w.r.t.\ $\eps$ and the fact that $\To{\eps}$ explains
  $\eps$, implies that $(\Ts,\ls)$ is isomorphic to $\To{\eps}$.
	
  Moreover, let $(T,\lambda)$ be an edge-labeled tree that explains
  $\eps$. Then, by \cref{prop:coarse-graining},
  $\To{\eps}\leq (T,\lambda)$.  By definition of ``coarse-graining'', we
  have $\Cls(\wh{T})\subseteq\Cls(T)$, and hence
  $|\Cls(\wh{T})|\leq|\Cls(T)|$.  Since $\X \in \Cls(\wh{T})\cap \Cls(T)$,
  and since $\{x\}\in \Cls(\wh{T})\cap \Cls(T)$ for all $x\in \X$, we can
  conclude that $|V(\wh{T})| \leq |V(T)|$.  Since the latter is satisfied
  for every edge-labeled tree $(T,\lambda)$ that explains $\eps$, the tree
  $\To{\eps}$ must have a minimum number of vertices.

  Furthermore, by definition of ``coarse-graining'', we also have
  $\wh{\lambda}(\parent\lc{\wh{T}}(\hat{v}),\hat{v})\subseteq
  \lambda(\parent\lc{T}(v),v)$
  for all $\hat{v} \in V(\wh{T})\setminus\{\rootT{\wh{T}}\}$ and for all
  $v \in V(T)\setminus\{\rootT{T}\}$ with $\Cl{\wh{T}}(\hat{v})=\Cl{T}(v)$.
  Hence, the sum $\sum_{e\in E(\wh{T})}|\wh{\lambda}(e)|$ is minimum among
  all edge-labeled trees that explains $\eps$.
\end{proof}

Note that \Cref{thm:uniqueLRT}, together with \cref{prop:coarse-graining},
is equivalent to \cite[Thm.\ 1]{Geiss:18a} in case $\eps$ is a
monochromatic Fitch map.  Moreover, \Cref{thm:uniqueLRT} implies that one
can verify whether a tree $(T,\lambda)$ is least-resolved w.r.t.\ a Fitch
map $\eps$ by checking if $(T,\lambda)$ is isomorphic to $\To{\eps}$.  An
alternative way to test if a tree is least-resolved is provided by the next
result.
\begin{proposition}
  Let $\eXM$ be a Fitch map that is explained by $(T, \lambda)$, and let
  $\To{\eps}$ be the \nameT.  Then, $(T,\lambda)$ is isomorphic to
  $\To{\eps}$ if and only if for every inner edge
  $e=(\parent\lc{T}(v),v)\in \mathring{E}(T)$ the following two statements
  are satisfied:
  \begin{enumerate}[noitemsep]
  \item $\lambda(e)\ne \emptyset$, and
  \item for every $m\in \lambda(\parent(v),v)$ there is a leaf $y\in \X$
    such that there is no $m$-edge along the path from $v$ to $y$ in
    $(T,\lambda)$.
  \end{enumerate}
  \label{prop:alternative}
\end{proposition}
\begin{proof}
  Let $\eXM$ be a Fitch map that is explained by $(T, \lambda)$, and let
  $\To{\eps}=(\wh{T},\wh{\lambda})$ be the \nameT.
	
  First, assume that $(T,\lambda)$ is isomorphic to $\To{\eps}$, and thus
  $\Cls(T)=\Cls(\wh{T})$.  Then, we may assume w.l.o.g.\ that
  $V(T)=V(\wh{T})$ and $E(T)=E(\wh{T})$.  Hence, by
  \cref{def:tree-for-eps}, for all $v\in V(T)\setminus \{\rootT{T}\}$ it
  holds that $\Cl{T}(v)=\Cl{\wh{T}}(v) = \NnotCol{m}[y]$ for some $y\in \X$
  and some $m\in M$.  Now, we can utilize the equivalence between (1)
    and (2) in \Cref{p:clust-neighb} to conclude that Statement~(1) and
  (2) are satisfied for $(T,\lambda)$.

  Conversely, assume that $(T,\lambda)$ satisfies Statement~(1) and (2).
  Let $v \in \mathring{V}(T)\setminus\{\rootT{T}\}$ be an arbitrary inner
  vertex. Statement~(1) implies that there is an
  $m \in \lambda(\parent\lc{T}(v),v)$, and Statement~(2) implies that there
  is no $m$-edge along the path from $v$ to $y$ in $(T,\lambda)$. Now, we
  can apply \Cref{p:clust-neighb} to conclude that
  $\Cl{T}(v)=\NnotCol{m}[y]$ for some $y\in \X$ and some $m\in M$.  Hence,
  by \cref{def:tree-for-eps}~(a), we have
  $\Cl{T}(v)=\NnotCol{m}[y] \in \Ns[\eps]\subseteq\Ns[\eps]\cup\{ \{x\} :
  x\in \X \}\cup\{ \X \}=\Cls(\wh{T})$. This implies
  $\Cls(T)\subseteq\Cls(\wh{T})$. Moreover, \cref{prop:coarse-graining}
  implies that $\Cls(\wh{T})\subseteq\Cls(T)$. Hence, we obtain
  $\Cls(\wh{T})=\Cls(T)$.
				
  Now, we need to show that
  $\lambda(\parent(v),v)=\wh{\lambda}(\parent\lc{\wh{T}}(\hat{v}),\hat{v})$
  is satisfied for every $v \in V(T)\setminus\{\rootT{T} \}$ and for every
  $\hat{v}\in V(\wh{T})$ with $\Cl{T}(v)=\Cl{\wh{T}}(\hat{v})$.
  \cref{prop:coarse-graining} implies that $(\wh{T},\wh{\lambda})$ is a
  coarse-graining of $(T,\lambda)$ and therefore,
  $\wh{\lambda}(\parent\lc{\wh{T}}(\hat{v}),\hat{v})\subseteq
  \lambda(\parent\lc{T}(v),v)$.  To verify that
  $\lambda(\parent\lc{T}(v),v)\subseteq
  \wh{\lambda}(\parent\lc{\wh{T}}(\hat{v}),\hat{v})$, let
  $m \in \lambda(\parent\lc{T}(v),v)$.  If $v \notin \X$ is not a leaf,
  then Statement~(2) implies that there is a leaf $y \in \X$ such that
  there is no $m$-edge on the path from $v$ to $y$ in $(T,\lambda)$. If
  $v \in \X$ is a leaf, then there is trivially no $m$-edge on the path
  from $v$ to $v$ in $(T,\lambda)$.  In both cases, \Cref{p:clust-neighb}
  implies that $\Cl{T}(v)=\NnotCol{m}[y]$.  Since $\Cls(\wh{T})=\Cls(T)$,
  as shown above, there is a vertex
  $\hat{v}\in V(\wh{T})\setminus \{\rootT{\wh{T}}\}$ such that
  $\Cl{\wh{T}}(\hat{v})=\Cl{T}(v)=\NnotCol{m}[y]$.  By
  \cref{def:tree-for-eps},
  $m\in \lambda(\parent\lc{\wh{T}}(\hat{v}),\hat{v})$ and thus
  $\wh{\lambda}(\parent\lc{\wh{T}}(\hat{v}),\hat{v}) =
  \lambda(\parent\lc{T}(v),v)$.

  In summary, $(T,\lambda)$ is isomorphic to $\To{\eps}$.
\end{proof}

\section{Restricted Fitch maps}
\label{sec:simple}

In the following, we consider two typical restrictions of Fitch maps.
  One restricts the number of colors placed on the edges and the other is
  based on a ``recoloring'' based on subsets of the color set $M$.

\subsection{$\boldsymbol{k}$-Restricted Fitch map}

\citet{Geiss:18a} considered monochromatic Fitch maps and
\citet{Hellmuth:2019d} considered disjoint Fitch maps.  These special
classes of Fitch maps can always be explained by edge-labeled trees
$(T,\lambda)$ with $|\lambda(e)|\leq 1$ for every $e \in E(T)$.  In view of
these results, we consider here a common generalization of these ideas
and ask which type of Fitch maps can be explained by edge-labeled trees
$(T,\lambda)$ with $|\lambda(e)|\leq k$ for every $e \in E(T)$ and some
fixed integer $k$.

\begin{definition}
  Let $\eXM$ be a Fitch map, and let $k\in \mathbb{N}$ be an integer.
  Then, we call $\eps$ a \emph{$k$-restricted} Fitch map if there is an
  edge-labeled tree $(T,\lambda)$ that explains $\eps$ and that satisfies
  $|\lambda(e)|\le k$ for every $e \in E(T)$.
  \label{def:k-restricted}
\end{definition}
Note that every monochromatic Fitch map and every disjoint Fitch map is a
$1$-restricted Fitch map.  In order to characterize $k$-restricted Fitch
maps, we will use the \nameT $\To{\eps}$, \Cref{prop:coarse-graining}
and the following
\begin{definition} \label{def:k-ELC} A map $\eXM$ satisfies the
  $k$-edge-label-condition ($k$-\EC) if for every neighborhood
  $N\in \Ns[\eps]$ with $N\ne \X$ we have
  $\big|\{ m \in M \colon \textnormal{there is a } y \in \X
  \textnormal{ with }
  N = \NnotCol{m}[y]\}\big| \le k$.
\end{definition}
In other words, $\eps$ satisfies $k$-\EC if for every neighborhood
$N\in \Ns[\eps]$ with $N\ne \X$ there are at most $k$ colors in $M$ for
which $N = \NnotCol{m}[y]$ is satisfied.

\begin{proposition}\label{prop:simple-fitch}
  Let $\eXM$ be a Fitch map, and let $\To{\eps}=(\wh{T},\wh{\lambda})$ be
  the \nameT.  Then, the following three statements are equivalent:
  \begin{enumerate}[noitemsep]
  \item $\eps$ is a $k$-restricted Fitch map.
  \item For every edge $e \in E(\wh{T})$ we have $|\wh{\lambda}(e)|\le k$.
  \item $\eps$ satisfies $k$-\EC.
  \end{enumerate}
\end{proposition}
\begin{proof}
  Let $\eXM$ be a Fitch map, and let $\To{\eps}=(\wh{T},\wh{\lambda})$ be
  the \nameT.
	
  First, suppose that Statement~(1) is satisfied. Then, there is an
  edge-labeled tree $(T,\lambda)$ with $|\lambda(e)|\le k$ for every edge
  $e \in E(T)$.  By \cref{prop:coarse-graining}, the tree
  $\To{\eps}=(\wh{T},\wh{\lambda})$ is a coarse-graining of $(T,\lambda)$,
  i.e.\ $\Cls(\wh{T})\subseteq \Cls(T)$ and for each
  $\hat{v} \in V(\wh{T})\setminus\{\rootT{\wh{T}}\}$ and for each
  $v \in V(T)\setminus\{\rootT{T}\}$ with $\Cl{T}(v)=\Cl{\wh{T}}(\hat{v})$
  we have
  $\wh{\lambda}(\parent\lc{\wh{T}}(\hat{v}),\hat{v})\subseteq
  \lambda(\parent\lc{T}(v),v)$.
  This, together with $|\lambda(e)|\le k$ for every $e \in E(T)$,
  immediately implies Statement~(2).
	
  Next, assume that Statement~(2) is satisfied. By \Cref{lem:sufficiency},
  $\To{\eps}$ explains $\eps$.  Thus, by \Cref{def:k-restricted},
  Statement~(1) is trivially satisfied. Hence, Statement~(1) and (2) are
  equivalent.
 
  We are still assuming that Statement~(2) is satisfied, and let
  $N\in \Ns[\eps]$ with $N\ne \X$.  Then, \cref{def:tree-for-eps}~(a)
  implies that $N \in \Ns[\eps] \subseteq \Cls(\wh{T})$ is a cluster.
  This, together with $N\ne\X$, implies that there is a vertex
  $\hat{v}\in V(\wh{T})$ with $\hat{v}\ne\rootT{\wh{T}}$ such that
  $\Cl{\wh{T}}(\hat{v})=N$.  Since $\hat{v}\ne \rootT{\wh{T}}$, there is
  the edge $(\parent\lc{\wh{T}}(\hat{v}),\hat{v}) \in E(\wh{T})$.  Then,
  \Cref{def:tree-for-eps}~(b), together with $\Cl{\wh{T}}(\hat{v})=N$ and
  Statement~(2), implies that
  \begin{equation*}
    \big|\{ m \in M \colon \textnormal{ there is a } y\in \X
    \textnormal{ with }N= \Cl{\wh{T}}(\hat{v}) =\NnotCol{m}[y]\} \big|=
    \big|\wh{\lambda}(\parent\lc{\wh{T}}(\hat{v}),\hat{v})\big| \le k.
  \end{equation*}
  Hence, Statement~(3) is satisfied.
	
  Now, assume that Statement~(3) is satisfied.  Then, by
  \Cref{def:tree-for-eps}~(b), we immediately conclude that Statement~(2)
  holds.  Hence, Statement~(2) and (3) are equivalent.
\end{proof}

\Cref{prop:simple-fitch}, together with \Cref{thm:charact-HCIC}, implies
\begin{theorem}\label{thm:HCIC-ELC}
  A map $\eXM$ is a $k$-restricted Fitch map if and only if $\eps$
  satisfies \HC, \IC and $k$-\EC.
\end{theorem}

\begin{figure}[t]
  \centering \includegraphics[width=0.85\textwidth]{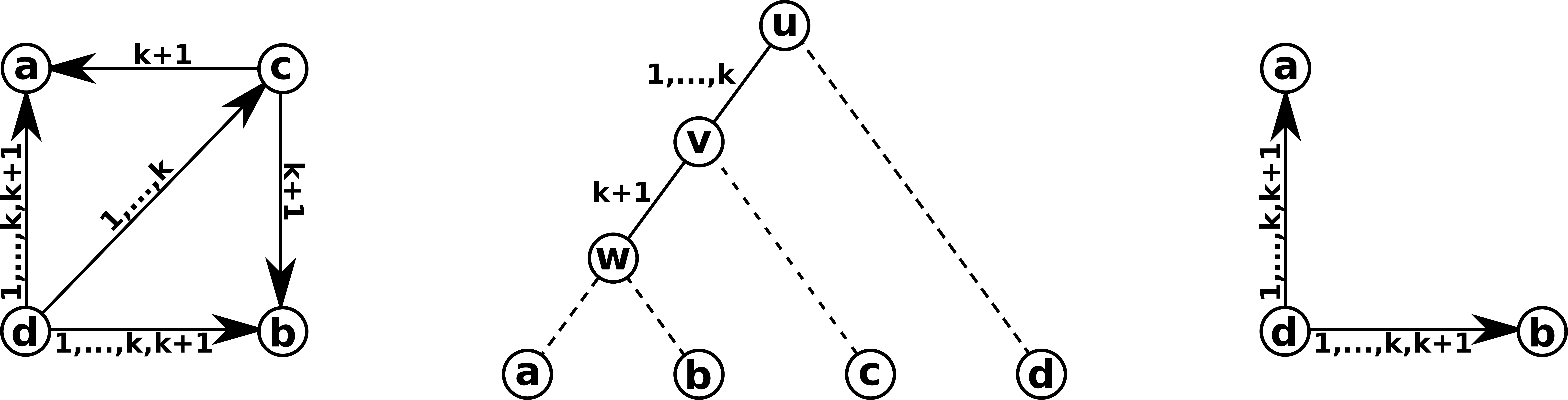}
  \caption{An example that shows that $k$-restricted Fitch maps can contain
    induced submaps that are not $k$-restricted Fitch maps, see
    \Cref{exmpl:k-res} for further explanations.}
  \label{fig:k-restricted}
\end{figure}

Thus, we were able to adjust the characterization as in
\Cref{thm:charact-HCIC} for the special class of $k$-restricted Fitch maps.
However, as we shall see later, it is not possible to derive a
characterization in terms of forbidden submaps similar to
\cref{thm:forbid-sub}.  To this end, consider first the following

\begin{exmpl}[$k$-restricted Fitch maps may contain induced submaps
  that are not $k$-restricted Fitch maps]
  \label{exmpl:k-res}
  Consider the map $\eXM$ with $\X=\{a,b,c,d\}$ and $M=\{1,\ldots,k,k+1\}$,
  $k \geq 1$ as shown in \cref{fig:k-restricted}~(left).  For the
  edge-labeled tree $(T,\lambda)$ as provided in
  \cref{fig:k-restricted}~(middle), all dashed edges $e\in E(T)$ have label
  $\lambda(e)=\emptyset$.  Hence, $|\lambda(e)|\leq k$ for all $e\in
  E(T)$. In fact, $(T,\lambda)$ explains $\eps$; and therefore, $\eps$ is a
  $k$-restricted Fitch map.
	
  Now, consider the induced submap
  $\eps'\colon \irr{X'\times X'}\to \PS{M}$ with $X'=\{a,b,d\}$ of $\eps$,
  as shown in \cref{fig:k-restricted}~(right).  By
  \cref{cor:induced-Fitch}, $\eps'$ is also a Fitch map.  For
  $N\coloneqq \{a,b\}=\NnotCol{1}[a]\in \Ns[\eps']$ we have
  $|\{ m \in M \colon N = \NnotCol{m}[a] \}|=|M|=k+1>k$.  This, together
  with $\{a,b\}\ne \X'$, implies that $\eps'$ does not satisfy the $k$-\EC,
  cf.\ \cref{def:k-ELC}.  This, the fact that $\eps'$ is a Fitch map, and
  \cref{prop:simple-fitch} imply that $\eps'$ is not a $k$-restricted Fitch
  map.
	
  In summary, although $\eps$ is a $k$-restricted Fitch map, it contains an
  induced submap, which is not a $k$-restricted Fitch map.
\end{exmpl}

Since we only consider phylogenetic trees, we can utilize
  \Cref{exmpl:k-res} in order to show that there is no characterization of
  $k$-restricted Fitch maps in terms of a set of forbidden submaps. This
  statement can expressed more formally as follows:
\begin{theorem}\label{thm:no-submap-charct}
  There is no set of forbidden submaps such that $\eps$ is a $k$-restricted
  Fitch map if and only if $\eps$ does not contain a forbidden submap.
\end{theorem}
\begin{proof}
  Assume for contradiction that there is a set of forbidden submaps that
  characterizes $k$-restricted Fitch maps.  Let $\eps$ and $\eps'$ be
  chosen as in \Cref{exmpl:k-res}.  Since $\eps$ is a $k$-restricted Fitch
  map, $\eps$ does not contain any of such forbidden submaps.  Hence, the
  induced submap $\eps'$ of $\eps$ cannot contain any of these forbidden
  submaps.  Thus, $\eps'$ must be a $k$-restricted Fitch map; a
  contradiction.
\end{proof}

In principle it is possible to relax the restriction to phylogenetic
  trees. Allowing arbitrary trees, we would obtain the same
  characterization as for Fitch maps, i.e.\
  \cref{thm:charact-HCIC,thm:forbid-sub}. To see this, it suffices to
  replace every edge $e$ of a phylogenetic tree $T$ by a path of sufficient
  lengths to assign at most $k$ distinct colors from $\lambda(e)$ to each of
  the edges in the subdivision. Moreover, we note in passing that there
are forbidden submap characterizations for highly constrained
$1$-restricted Fitch maps such as monochromatic Fitch maps \cite{Geiss:18a}
or disjoint Fitch maps \cite{Hellmuth:2019d}.

\subsection{Recoloring Fitch Maps}

Interpreting colors as different subclasses of horizontal transfer events
it is of interest to consider different resolutions at which events are
considered different. Considering certain sets of colors as equivalent thus
amounts to a course graining. Here, we briefly show that Fitch maps are
well-behaved under ``recoloring'' and ``identification of colors''.

\begin{definition}
  Let $\eXM$ be a map and $P=\{M_1,\ldots,M_k\}\subseteq\PS{M}$ be a
  collection of subsets of $M$.  Then, we define the \emph{\Pcol map}
  $\eps\lc{P}\colon \irr{\X\times \X} \to \PS{\{1,\ldots,k\}}$ by putting
  for all distinct $x,y \in \X$
  \begin{equation*}
    \eps\lc{P}(x,y) \coloneqq \big\{ i \in \{1,\ldots,k\} : \eps(x,y) \cap
    M_i \ne \emptyset \big\}.
  \end{equation*}
\end{definition}
In other words, $\eps\lc{P}(x,y)$ contains color $i$ if and only if
$\eps(x,y)\cap M_i\ne\emptyset$. Given a Fitch map $\eXM$ explained by
  the edge-labeled tree $(T,\lambda)$ and given a set
  $P=\{M_1,\ldots,M_k\}\subseteq\PS{M}$, we will make use of the
  edge-labeled tree $(T,\lambda\lc{P})$, where the edge-labeling
  $\lambda\lc{P} : E(T)\to \PS{\{1,\ldots,k\}}$ assigns the set
  $\lambda\lc{P}(e) \coloneqq \big\{ i \in \{1,\ldots,k\} : \lambda(e)\cap
  M_i\ne \emptyset \big\}$ to every edge $e\in E(T)$.

\begin{proposition} \label{prop:re-coloring} Let $\eXM$ be a Fitch map, and
  let $P=\{M_1,\ldots,M_k\}\subseteq\PS{M}$ be a collection of subsets of
  $M$.  Then, the \Pcol map $\eps\lc{P}$ is a Fitch map that is explained
  by $(T,\lambda\lc{P})$.
\end{proposition}
\begin{proof}
  Suppose $\eXM$ is a Fitch map explained by $(T,\lambda)$, and let
  $P=\{M_1,\ldots,M_k\}\subseteq\PS{M}$.  Since $(T,\lambda)$ explains
  $\eps$, we have for every $i \in \{1,\ldots,k\}$ and for every distinct
  $x,y \in \X$:
  \begin{align*}
    i \in \eps\lc{P}(x,y) & \iff \textnormal{there is an }
                            m \in \eps(x,y)\cap M_i \\
                          &\iff \textnormal{there is an edge } e \in P\lc{T}(\lca(x,y),y) \textnormal{ with } m \in \lambda(e)\cap M_i \\
                          &\iff \textnormal{there is an edge } e \in P\lc{T}(\lca(x,y),y) \textnormal{ with } i \in \lambda\lc{P}(e).
  \end{align*}
  Since we have chosen $i \in \{1,\ldots,k\}$ and $x,y \in \X$ arbitrarily,
  we conclude that $(T,\lambda\lc{P})$ explains $\eps\lc{P}$; and thus,
  that $\eps\lc{P}$ is a Fitch map.
\end{proof}
Note that \Cref{thm:uniqueLRT} implies that $\To{\eps\lc{P}}$ is a
coarse-graining of $(T,\lambda\lc{P})$ for every \Pcol Fitch map
$\eps\lc{P}$.  In particular, \Cref{prop:re-coloring} allows us to
identify colors. In this case $P= \{M_1,\ldots,M_k\}$ is a partition of
$M$. Thus we have
\begin{corollary}
  Let $\eXM$ be a Fitch map and $P=\{M_1,\ldots,M_k\}$, $k\geq 1$ be a
  partition of $M$.  Then, the \Pcol map $\eps\lc{P}$ is a Fitch map that is
  explained by $(T,\lambda\lc{P})$.
\end{corollary}
In particular, therefore, the least resolved tree explaining $\eps\lc{P}$
is displayed by the least resolved tree explaining $\eps$.  Finally, we
note that $\eps\lc{P}=\eps$ whenever $P$ consists of all singletons contained in $\PS{M}$.

\section{Algorithmic Considerations}
\label{sec:algo}

\Cref{alg:fitch} summarizes a method to recognize Fitch maps and, in the
affirmative case, to construct the corresponding (unique) least-resolved
edge-labeled tree. In this algorithm it must be verified whether the
computed set $\Ns[\eps]$ forms a hierarchy or not.  Although there are
  papers that implicitly use algorithms to test whether a set system is
hierarchy-like or not based e.g.\   on underlying Hasse diagrams
\cite[Section 5]{GCG+18} or so-called character-compatibility
  \cite[Section 7.2]{sung2009algorithms}, we provide here a quite simple
alternative direct algorithm (cf.\ \cref{alg:test_hierarchy}).

\begin{lemma}
  Given a collection $\mc{C}\subseteq \PS{\X}$ of subsets of $\X$,
  \cref{alg:test_hierarchy} correctly determines whether $\mc{C}$ is
  hierarchy-like or not in
  $\mc{O}(|\mc{C}||\X|)\subseteq \mc{O}(|\X|^2)$ time.
  \label{lem:test_hierarchy}
\end{lemma}
\begin{proof}
  First, we prove the correctness of the algorithm. Let $\X$ be a non-empty
  set, and let $\mc{C}=\{C_1,C_2,\ldots,C_{|\mc{C}|}\}\subseteq \PS{\X}$ be
  a collection of subsets of $\X$. By \cite[Lemma 1]{HLS+15}, if $\mc{C}$
  is a hierarchy, then $|\mc{C}|\leq 2|X|-1$. Hence, if $|\mc{C}|> 2|X|-1$,
  then $\mc{C}$ cannot be a hierarchy, and thus $\mc{C}$ cannot be
  hierarchy-like. In this case, the algorithm correctly returns
  \texttt{false} in \Cref{alg:card}.

  Then, the set $\mc{C}$ is ordered based on the cardinality of its
  elements (\Cref{alg:sort}). Moreover, a map
  $\varphi\colon \X\to \{0,1,\dots,|\mc{C}|\}$ is initialized with
  $\varphi(x)=0$ for every $x \in \X$ (\Cref{alg:beginxyz}).  In essence,
  $\varphi(x)$ saves for each $x \in \X$ the last considered set $C_i$
  where $x$ was discovered. The initial case $\varphi(x)=0$ corresponds to
  the trivial case ``$x\in C_0=\X$'' in the subsequent parts of this proof.
		
  \Crefrange{alg:i-th-step}{alg:endxyz} iterates over all
  $C_i\in\mc{C}$ from the largest to the smallest elements.  We set
  $j\gets\varphi(x)$ for some arbitrary but fixed vertex $x\in C_i$. Thus,
  $j$ is now the index of the latest preceding set $C_j$ that contains
  $x$.  Then, we check for all $y \in C_i$ whether index $\varphi(y)=j$,
  that is, whether $y\in C_j$ is true for all all $y \in C_i$, i.e.,
  whether $C_i\subseteq C_j$.  If this is the case, then the value
  $\varphi(y)$ is changed to the current index $i$; otherwise, the
  algorithm returns \texttt{false}.
 
  It remains to show that \cref{alg:test_hierarchy}
  (\Crefrange{alg:i-th-step}{alg:endxyz}) returns \texttt{false} if and
  only if $\mc{C}$ is not hierarchy-like. First, suppose that
  \cref{alg:test_hierarchy} returns \texttt{false}, which is the case if
  there are vertices $x,y \in C_i$ that satisfy
  $\varphi(x)=j\ne \varphi(y)$. Hence, $x \in C_i\cap C_{\varphi(x)}$ and
  $y \in C_i\cap C_{\varphi(y)}$. Note that $\varphi(x),\varphi(y)<i$ and
  we may assume w.l.o.g.\ that $\varphi(x)<\varphi(y)$. Thus,
  $x\notin C_{\varphi(y)}$, as otherwise, the value $\varphi(x)$ must have
  been changed to the index $\varphi(y)$ when considering $C_{\varphi(y)}$,
  since $C_{\varphi(y)}$ is considered after $C_{\varphi(x)}$.  However, in
  this case, $|C_i|\le|C_{\varphi(y)}|$,
  $y \in C_i\cap C_{\varphi(y)}\ne \emptyset$ and
  $x \in C_i\setminus C_{\varphi(y)}$ implies
  $C_i\cap C_{\varphi(y)} \notin
  \{\emptyset,C_i,C_{\varphi(y)}\}$. Therefore, $\mc{C}$ is not
  hierarchy-like.

  Conversely, suppose that $\mc{C}$ is not hierarchy-like.  Then, there are
  two elements $C_i,C_j \in \mc{C}$ such that
  $C_i\cap C_j \notin \{\emptyset, C_i,C_j\}$.  In particular, we can
  choose the indices $i$ and $j$ such $j<i$ and $i-j$ is minimum.  Now,
  consider the step of the algorithm where $C_i$ is investigated.  Then, we
  may assume w.l.o.g.\ that $x \in C_i\cap C_j$ and
  $y \in C_i\setminus C_j$.  Hence, by the choice of $i$ and $j$, we
  conclude that $\varphi(x)=j$, and since $y \notin C_j$, we conclude that
  $\varphi(y)\ne j$.  Thus, $\varphi(x) \ne \varphi(y)$ and the
  \emph{if}-condition in \Cref{alg:setvarphi} correctly will return
  \texttt{false}. In summary, \cref{alg:test_hierarchy} returns
  \texttt{false} if and only if $\mc{C}$ is not hierarchy-like.

  We continue by investigating the running time of the algorithm. Due to
  the \emph{if}-condition in \Cref{alg:card}, we can observe that
  $|\mc{C}|\in \mc{O}(|\X|)$.  Thus, the sorting of the elements in
  $\mc{C}$ (\Cref{alg:sort}) can be achieved in $\mc{O}(|\X|\log(|X|))$
  time.  Moreover, we iterate in \Crefrange{alg:i-th-step}{alg:endxyz} over
  all $|\mc{C}|\in \mc{O}(|\X|)$ elements in $\mc{C}$ and all
  $\mc{O}(|\X|)$ elements in each $C_i\in \mc{C}$ ending an overall time
  complexity of $\mc{O}(|\mc{C}||\X|)\subseteq\mc{O}(|\X|^2)$, which
  completes the proof.
\end{proof}

\begin{algorithm}[t]
  \caption{\texttt{Test whether a set $\mc{C}\subseteq \PS{\X}$ is
      hierarchy-like}}
  \begin{algorithmic}[1]
  \Require  $\mc{C}\subseteq \PS{\X}$.
  \Ensure  \texttt{true}, if $\mc{C}$ is hierarchy-like, and
     \texttt{false}, otherwise.
  \If{$|\mc{C}|>2|X|-1$} \Return \texttt{false} \label{alg:card} \EndIf 
  \State Order $\mc{C} = \{C_1,\dots C_{|\mc{C}|}\}$ such that $i \le j$
     whenever $|C_i|\geq |C_j|$ \label{alg:sort}
  \State $\varphi(x)\gets 0$ for all $x\in \X$\label{alg:beginxyz}
  \Comment{Construct a map $\varphi\colon \X\to \{0,1,\dots,|\mc{C}|\}$ }
  \ForAll{$i \in \{1,\dots, |\mc{C}|\}$}\label{alg:i-th-step}
     \State $j \gets \varphi(x)$ for some arbitrary $x \in C_i$
        \label{alg:choose-x} 
     \ForAll{$y\in C_i$} \label{alg:for-loob-Ci} 
	\If{$\varphi(y) = j$} $\varphi(y)\gets i$ \label{alg:setvarphi}
	\Else\ \Return \texttt{false} \label{alg:endxyz} 
	\EndIf
     \EndFor
  \EndFor
  \State \Return \texttt{true}
  \end{algorithmic}
  \label{alg:test_hierarchy}
\end{algorithm}
%%%%%%%%%%%%%%%%%%%%%%%%%%%%%%%%%%%%%%%%%%%
\begin{algorithm}[t]
  \caption{\texttt{Determining Fitch Maps and Computing the Least-Resolved
      Tree $\To{\eps}$ for $\eps$}}
  \begin{algorithmic}[1]
    \Require A map $\eXM$, where $\eps(x,y)$ is an ordered set for
        all $(x,y)\in \irr{\X \times \X}$.
    \Ensure  Least-resolved tree that explains $\eps$, if one exists.
    \State $\NnotCol{m}[y]\gets\emptyset$ for all $y\in \X$ and all
        $m\in M$\Comment{Compute $\NnotCol{m}[y]$}\label{alg:for1-begin} 	
    \ForAll{$y\in \X$, and \textbf{for all} $x\in \X\setminus\{y\}$}
        \label{alg:forxy}
        \ForAll{$m\in M$ with $m\notin \eps(x,y)$}
            \State $\NnotCol{m}[y]\gets \NnotCol{m}[y]\cup \{x\}$
        \EndFor
    \EndFor
    \State $\NnotCol{m}[y]\gets \NnotCol{m}[y]\cup \{y\}$ for all
        $y\in \X$ and all $m\in M$\label{alg:for1-end}
    \State $\Ns[\eps]\gets \emptyset$, and $\textrm{count}[\ell]\gets0$
        for all $\ell\in\{1,\dots,|\X|\}$
        \Comment{Compute $\Ns[\eps]$} \label{alg:Nseps}
    \ForAll{$m\in M$, and \textbf{for all} $y\in \X$} \label{alg:forNeps}
        \If{$\NnotCol{m}[y] \notin \Ns[\eps]$} \label{alg:Nmy-in-Neps}
            \State  $\Ns[\eps]\gets  \Ns[\eps]\cup \{\NnotCol{m}[y]\}$
                \label{alg:unioin-neigh}
	    \State $\textrm{label}[m,y]\gets\{m\}$ \label{alg:set-label}
	    \State $\textrm{count}[|\NnotCol{m}[y]|]\gets
                \textrm{count}[|\NnotCol{m}[y]|]+1$
            \If{$|\mc{\Ns[\eps]}|>2|X|-1$ or
                $\textrm{count}[|\NnotCol{m}[y]|]\cdot |\NnotCol{m}[y]| >
                |X|$} \label{alg:value-compare}
	        \Return ``$\eps$ is not a Fitch map'' \label{alg:cancel-fitch}
	    \EndIf
        \Else
            \State Let $\NnotCol{m'}[y']\in \Ns[\eps]$ be the unique element
                with $\NnotCol{m}[y] = \NnotCol{m'}[y']$
                \label{alg:unique-neigh}
            \State $\textrm{label}[m',y'] \gets
                \textrm{label}[m',y'] \cup \{m\}$ \label{alg:labelunion}
        \EndIf
    \EndFor \label{alg:EndforNeps}
    \If{$\Ns[\eps]$ does not satisfy \IC or \HC}\label{alg:ifICHC}
        \Return ``$\eps$ is not a Fitch map'' \label{alg:ICHC}
        \Comment{Verify \IC and \HC}
    \EndIf
    \State Compute  $T$ with $\Cls(T) = 
        \Ns[\eps] \cup \big\{\X\big\} \cup
        \big\{ \{x\} \colon x \in \X \big\}$ \Comment{Compute $\To{\eps}$}
        \label{alg:To-begin}
    \ForAll{edges $(\parent(v),v)$ of $T$} \label{alg:for-edges}
  	\If{there is a (unique) element $\NnotCol{m}[y]\in \Ns[\eps]$
              with $\Cl{T}(v) =\NnotCol{m}[y]$} 
  	    \State $\lambda(\parent(v),v) \gets \textrm{label}[m,y]$
  	\Else{} $\lambda(\parent(v),v) \gets \emptyset$ 
  	\EndIf
     \EndFor \label{alg:To-end}
     \State \Return $\To{\eps} = (T,\lambda)$
  \end{algorithmic}
  \label{alg:fitch}
\end{algorithm}

 Let us now consider \Cref{alg:fitch} that determines whether a given
  map $\eXM$ is a Fitch map or not, and that returns, in the affirmative
  case, the least-resolved tree that explains $\eps$.  We shall note first
  that many of the more elaborate parts of the algorithm are used to
  achieve the desired running time.  In a nutshell, in
  \crefrange{alg:for1-begin}{alg:EndforNeps} the collection of
  neighborhoods $\Ns[\eps]$ is computed.  In addition, an array of sets, 
	called $\textrm{label}$, is computed where $\textrm{label}[m,y]$ contains all
  labels that need to be added on particular edges of the possible existing
  tree that explains $\eps$.  Moreover, $\textrm{count}[|\NnotCol{m}[y]|]$
  is the number of elements in $\Ns[\eps]$ that have the same cardinality
  as $\NnotCol{m}[y]$. If $\textrm{count}[|\NnotCol{m}[y]|]$ is larger than
  some specified values, then we can directly verify that $\eps$ is not a Fitch
  map (\Cref{alg:value-compare}). Then, we continue to check in
  \Cref{alg:ifICHC} if $\Ns[\eps]$ satisfies \IC and \HC.  In the
  affirmative case, \Cref{thm:charact-HCIC} implies that $\eps$ is a Fitch
  map, and we compute in \crefrange{alg:To-begin}{alg:To-end} the unique
  least-resolved tree $\To{\eps}$ that explains $\eps$.  

\begin{theorem}
  \Cref{alg:fitch} determines correctly whether a given map $\eXM$ is a
  Fitch map. In the affirmative case \Cref{alg:fitch} returns the
  least-resolved tree that explains $\eps$.  \Cref{alg:fitch} can be
  implemented to run in $\mc{O}(|\X|^2\cdot |M|)$ time.
\label{thm:algo}
\end{theorem}
\begin{proof}
  Let $\eXM$ be a map.  First, we prove the correctness of the algorithm.
  It is easy to verify that the block consisting of the
  \crefrange{alg:for1-begin}{alg:for1-end} correctly computes
  $\NnotCol{m}[y]$ for all $y\in \X$ and $m\in M$.

  Now, we verify that the block consisting of the
  \crefrange{alg:Nseps}{alg:EndforNeps} correctly computes
  $\Ns[\eps]$. First, for all possible cardinalities
  $\ell\in\{1,\dots,|\X|\}$ a counter $\textrm{count}[\ell]=0$ is
  initialized, see \Cref{alg:Nseps}. This counter will count all
  neighborhoods that have the same size $\ell$.
	
  Then, we iterate over all $m\in M$, and over all $y\in \X$.  In
  \Cref{alg:Nmy-in-Neps}, we check if $\NnotCol{m}[y]$ is already contained
  in $\Ns[\eps]$ or not.  If $\NnotCol{m}[y]$ is not contained in
  $\Ns[\eps]$, then we add it to $\Ns[\eps]$. The latter, in particular,
  ensures that $\Ns[\eps]$ is a set and not a multi-set.  Moreover, we
  initialize an array of sets $\textrm{label}[m,y] = \{m\}$ that will
  contain all labels that we need to add on particular edges of the
  possible exiting tree that explains $\eps$.  Moreover, we increase
  $\textrm{count}[|\NnotCol{m}[y]|]$ by one, that is, we increment the
  number of elements in $\Ns[\eps]$ that have the same cardinality as
  $\NnotCol{m}[y]$.  In \Cref{alg:cancel-fitch}, we check if (i)
  $|\mc{\Ns[\eps]}|>2|X|-1$ or (ii)
  $\textrm{count}[|\NnotCol{m}[y]|]\cdot |\NnotCol{m}[y]| > |X|$ is
  satisfied. In Case (i) we can apply \cite[Lemma 1]{HLS+15} and conclude
  that $\Ns[\eps]$ cannot form a hierarchy. In this case, the algorithm
  correctly returns \emph{``$\eps$ is not a Fitch map''}.  In Case (ii), we
  check if the number of elements in $\Ns[\eps]$ that have the same
  cardinality as $\NnotCol{m}[y]$ times the number of elements in
  $\NnotCol{m}[y]$ exceeds $|X|$.  Suppose that Case (ii) applies. Since
  the elements of $\Ns[\eps]$ are pairwise distinct, they differ in at
  least one element. Thus, whenever there are two elements
  $N',N\in \Ns[\eps]$ of the same size then $N\cap N'=\emptyset$ if
  $\Ns[\eps]$ forms a hierarchy. But, then
  $\textrm{count}[|\NnotCol{m}[y]|]\cdot |\NnotCol{m}[y]|\leq X$, if
  $\Ns[\eps]$ forms a hierarchy. By contraposition, if Case (ii) applies,
  then $\Ns[\eps]$ cannot form a hierarchy. Hence, in both Cases (i) and
  (ii), the algorithm correctly returns \emph{``$\eps$ is not a Fitch
    map''}.

  If $\NnotCol{m}[y]$ is already contained in $\Ns[\eps]$
  (\emph{else}-case), then there is a unique neighborhood
  $\NnotCol{m'}[y'] \in \Ns[\eps]$ with $\NnotCol{m}[y]=\NnotCol{m'}[y']$.
  In this case, we simply save the particular color $m$, and add $m$ to the
  $\textrm{label}[m',y']$.  Clearly,
  $\Ns[\eps]=\{\NnotCol{m}[y] \colon y \in \X, m\in M\}$ is correctly
  computed.

  According to \cref{thm:charact-HCIC}, $\eps$ is a Fitch Map if and only
  if $\eps$ satisfies \IC and \HC.  Thus, the algorithm correctly returns
  \emph{``$\eps$ is not a Fitch map''}, in case $\eps$ does not satisfy \IC
  or \HC.  Hence, if $\eps$ is a Fitch map, then we can compute the
  edge-labeled tree $\To{\eps} = (T,\lambda)$ according to
  \cref{def:tree-for-eps}, which is done in
  \crefrange{alg:To-begin}{alg:To-end}.  In particular, for a vertex
  $v \in V(T)\setminus\{\rootT{T}\}$ the pre-computed set
  $\textrm{label}[m',y']$ consists of all colors $m$ for which there is a
  $y$ with $\Cl{T}(v) =\NnotCol{m}[y]$.  Hence, $\lambda(\parent(v),v) $ is
  correctly computed in \crefrange{alg:for-edges}{alg:To-end}.
  \cref{thm:uniqueLRT} states that $\To{\eps}$ is the least-resolved tree that
  explains $\eps$. In summary, \Cref{alg:fitch} is correct.

  Now, we investigate the running time of the algorithm.  To this end, we
  assume w.l.o.g.\ that $\X=\{1,\dots,|\X|\}$ and $M=\{1,\dots,|M|\}$ are
  the ordered sets of positive integers from $1$ to $|\X|$ and $1$ to
  $|M|$, respectively.  Moreover, as part of the input, $\eps(x,y)$ is an
  ordered set for all $(x,y)\in \irr{\X \times \X}$.

  Consider the block consisting of the
  \crefrange{alg:for1-begin}{alg:for1-end} that computes
  $\NnotCol{m}[y] = \{ x \in \X\setminus\{y\}\colon m \notin
  \eps(x,y)\}\cup \{y\}$.  First, we initialize 
  $\NnotCol{m}[y] \gets \emptyset$ for all $y\in\X$ and $m \in M$, a task that
  can be done in $\mc{O}(|X||M|)$ time.  Then, we iterate in
  \Cref{alg:forxy} over all $y\in\X$ and $x\in \X\setminus\{y\}$ in the
  order they appear in $\X$.  Then, we check for all $m\in M$ if
  $m\notin \eps(x,y)$ and, in the affirmative case, add $x$ to
  $\NnotCol{m}[y]$. Note that the resulting set $\NnotCol{m}[y]$ will then
  already be ordered.  To check if $m\notin \eps(x,y)$, we may first
  compute $\overline{M}_{xy}\coloneqq M\setminus \eps(x,y)$ in
  $\mc{O}(|M|)$ time and, afterwards, iterate over all elements
  $m \in \overline{M}_{xy}$, which are precisely those $m\in M$ with
  $m\notin \eps(x,y)$.  The computation of $\overline{M}_{xy}$ can be
  achieved in $\mc{O}(M)$ time for a fixed $y\in\X$ and
  $x\in \X\setminus\{y\}$.  Thus, the entire \emph{for}-loop in
  \Cref{alg:forxy} runs in $\mc{O}(|X|^2|M|)$ time.  Finally, we add $y$ to
  $\NnotCol{m}[y]$ for all $y\in \X$ and all $m\in M$ in such a way that
  $\NnotCol{m}[y]$ remains an ordered set which can be done in
  $\mc{O}(|X|)$ time for a fixed $y\in \X$ and $m\in M$.  Thus, the latter
  task can be achieved in $\mc{O}(|X|^2|M|)$ time for all $y\in \X$ and all
  $m\in M$.

  Next, consider the block consisting of the
  \crefrange{alg:forNeps}{alg:EndforNeps} where the set
  $\Ns[\eps] = \{\NnotCol{m}[x] \colon x\in \X,m \in M \}$ is computed.

  First, we show that the \emph{if}-condition in \Cref{alg:Nmy-in-Neps} can
  be computed in $\mc{O}(|\X|)$. Recall that the neighborhoods
  $\NnotCol{m}[y]$ are already ordered.  Hence, checking $\NnotCol{m}[y]=N$
  for some $N \in \Ns[\eps]$ can be done in $\mc{O}(|\NnotCol{m}[y]|)$ time.
  Clearly, we only need to verify $\NnotCol{m}[y]=N$ for those
  $N \in \Ns[\eps]$ with $|\NnotCol{m}[y]|=|N|$.  Hence, testing if
  $\NnotCol{m}[y]\notin \Ns[\eps]$, can be in done
  $\mc{O}(\textrm{count}[|\NnotCol{m}[y]|]\cdot |\NnotCol{m}[y]| +
  |\Ns[\eps]|)$ time.  Due to conditions in \Cref{alg:cancel-fitch},
  $\mc{O}(\textrm{count}[|\NnotCol{m}[y]|]\cdot
  |\NnotCol{m}[y]|+|\Ns[\eps]|) \subseteq O(|\X|)$.  Therefore, the
  \emph{if}-condition in \Cref{alg:Nmy-in-Neps} can be computed in time
  $\mc{O}(|\X|)$.
	
  \Cref{alg:unioin-neigh} can be evaluated in $\mc{O}(|\NnotCol{m}[y]|)$ time
  and the \crefrange{alg:set-label}{alg:cancel-fitch} require only
  constant time.  In \Cref{alg:unique-neigh}, finding the particular
  neighborhood $\NnotCol{m'}[y'] \in \Ns[\eps]$ with
  $\NnotCol{m}[y]=\NnotCol{m'}[y']$ has already been done in the
  \emph{if}-condition in \Cref{alg:Nmy-in-Neps}, and thus requires only the
  constant effort for a look-up.
     
  Since we iterate over all $m \in M$ in the given order of $M$, the set
  $\textrm{label}[m',y']$ is already sorted.  In particular, if $m$ is
  already contained in $\textrm{label}[m',y']$, then $m$ must be the last
  element of this set.  Hence, the union in \Cref{alg:labelunion} can be
  constructed in constant time.

  In summary, each individual step within the \emph{for}-loops in
  \Cref{alg:forNeps} can be accomplished in $\mc{O}(|\X|)$ time. Since the
  \emph{for}-loop has $\mc{O}(|\X||M|)$ iterations, we achieve a total
  running time of $\mc{O}(|\X|^2|M|)$ for the block consisting of the
  \crefrange{alg:forNeps}{alg:EndforNeps}.

  In \Cref{alg:ICHC}, we check if $\eps$ satisfies \IC and \HC.  To check
  $\HC$, we can use \cref{alg:test_hierarchy} to verify if $\Ns[\eps]$ is
  hierarchy-like in $\mc{O}(|\Ns[\eps]||\X|)$ time. Since we are
    in this step only if $|\Ns[\eps]|\leq 2|\X|-1\in \mc{O}(|\X|)$ (cf.\
    \Cref{alg:cancel-fitch}), the latter task requires $\mc{O}(|\X|^2|M|)$
    time.  To verify $\IC$, we check whether $|\NnotCol{m}[y']|\le|N|$ for
  every neighborhood $N\coloneqq \NnotCol{m}[y]$ with $m \in M$ and
  $y \in \X$ and for every $y' \in N$.  This requires $\mc{O}(|\X|^2|M|)$
  time.

  Finally, we compute $\To{\eps} = (T,\lambda)$ in
  \Crefrange{alg:To-begin}{alg:To-end}. To this end, we build $T$ based on
  the set
  $\Cls(T) = \Ns[\eps] \cup \big\{\X\big\} \cup \big\{ \{x\} \colon x \in
  \X \big\}$.  Since $\Ns[\eps]$ is hierarchy-like, we have
  $|\Ns[\eps]| \in \mc{O}(|X|)$. Hence, the hierarchy $\Cls(T)$ can be
  computed in $\mc{O}(|\X|)$ time.  Moreover, $T$ can be computed in
  $\mc{O}(|\Cls(T)|) \subseteq \mc{O}(|\X|)$ time, cf.\
  \cite{McConnell:05}.  While doing this, we also save the information,
  which vertex $v$ in $T$ corresponds to which cluster
  $\Cl{T}(v) =\NnotCol{m}[y]$.

  To compute $\lambda\colon E(T)\to\PS{M}$, we iterate over all
  $|V(T)|-1 \in \mc{O}(|\X|)$ edges in $T$ and check if there is a (unique)
  $\NnotCol{m}[y] \in \Ns[\eps]$ with $\Cl{T}(v) =\NnotCol{m}[y]$ in
  constant time based on the latter step, and then set
  $\lambda(\parent(v),v) =\textrm{label}[m,y]$.  The latter can
  be done in $\mc{O}(|M|)$ time for a fixed vertex $v$ in $T$.  Thus,
  computing $\lambda(\parent(v),v)$ for all
  $v \in V(T)\setminus\{\rootT{T}\}$ takes $\mc{O}(|\X||M|)$ time.

  In summary, \cref{alg:fitch} can be implemented to run in
  $\mc{O}(|\X|^2|M|)$ time.
\end{proof}

\section{Summary and Outlook}
\label{sec:outlook}

Fitch maps $\eps$ are map $\eXM$ that are explained by edge-labeled trees
$(T,\lambda)$ where $\lambda\colon E(T)\to\PS{M}$ assigns a subset of
colors in $M$ to each edge of $T$. The main result of this contribution is
to show that Fitch maps $\eps$ are characterized by the two simple
conditions \HC and \IC, which are both defined in terms of (complementary)
neighborhoods in the multi-edge-colored graph representation of $\eps$
(cf.\ \cref{thm:HCIC-ELC}).  Additionally, we provided a characterization
via forbidden submaps (cf.\ \cref{thm:forbid-sub}).  Moreover, we
demonstrated that there is a always a unique least-resolved tree for a
Fitch map (cf.\ \cref{thm:uniqueLRT}).  Finally, we gave a polynomial-time
algorithm to verify whether a given map $\eps$ is a Fitch map, and, in the
affirmative case, to construct the underlying least-resolved tree that
explains $\eps$.

Monochromatic Fitch maps have an additional characterization as a subclass
of so-called directed cographs \cite{Geiss:18a}.  This, in particular,
enables \citet{Geiss:18a} to establish a linear-time recognition algorithm
as well as a linear-time tree reconstruction method for monochromatic
Fitch maps. We suspect that there is a similar close relationship between
the Fitch maps defined here and so-called unp-2 structures
\cite{ER1:90,ER2:90,engelfriet1996characterization,Hellmuth:17a}, which
form a natural generalization of directed cographs to edge-colored graphs.
In particular, we expect that, using the theory of unp-2 structures,
the recognition of Fitch maps and the reconstruction of the least-resolved
trees is possible in a more efficient way.

As part of future research, it will be of interest to understand
symmetrized Fitch maps in more detail. A map $\eXM$ is a symmetrized Fitch
map if there is an edge-labeled tree $(T,\lambda)$ such that
$m\in \eps(x,y)$ if and only if there is an $m$-edge along the unique path
from $x$ to $y$ in $(T,\lambda)$.  A characterization of symmetrized
monochromatic Fitch maps can be found in \cite{GHLS:17}.  Note that both,
monochromatic Fitch maps and their symmetrized versions, form a special
subclass of (directed) cographs, which are graphs that can be explained by
vertex-labeled trees \cite{Geiss:18a,GHLS:17}.  A first attempt to
understand symmetrized Fitch maps can be found in \cite{Hellmuth:20}.  There,
a characterization in terms of quartets (unrooted phylogenetic trees on
four leaves) is provided, and it was shown that the recognition of
symmetrized Fitch maps is NP-complete.

The Fitch maps defined here correspond to directed multi-graphs with the
restriction that there are no parallel arcs of the same color. However, as
already outlined in \Cref{sec:char}, we may also allow parallel arcs of the
same color. That is, we may force to have $k$ $m$-edges along the path from
$\lca(x,y)$ to $y$, whenever there are $k$ edges with color $m$ connecting
$x$ and $y$ in the graph representation of $\eps$. To our knowledge, this
generalization has not been considered so far.

Finally, we have considered here only maps that are explained by trees.
Generalizations to maps that are defined by (vertex-labeled) networks can
be found in \cite{Huber2018b}.  Thus, a general question arises: Can a map
$\eXM$ that is not a Fitch map, and thus cannot be explained by an
edge-labeled tree, be explained by (rooted) edge-labeled networks instead?
What are the ``minimal'' or ``least-resolved'' networks that explain such a
map?

At present, there is no tool available to estimate the Fitch relation
  directly from data. There are, however, methods to retrieve partial
  information such as the fact that a given gene has undergone horizontal
  transfer \cite{Douglas:19}. On the one hand, the result reported here
  suggests that such partial information can help constrain gene trees.  On
  the other hand, our results show that Fitch maps contain a wealth of
  information on the gene tree and its embedding into the species tree,
  providing a strong motivation to investigate ways to infer them from data
  directly. In this context, it is interesting to note that the phylogenetic
  signal to infer the edge-labeled trees is entirely contained in the
  collection $\Ns[\eps]$ of ``complementary'' neighborhoods
  $\NnotCol{m}[y]$.  In other words, to reconstruct such trees it is not
  necessary to find those pairs $(x,y)$ for which some type of transfer
  happened, that is, $m\in \eps(x,y)$, but only to determine those pairs
  $(x,y)$ for which such an event has not occured, i.e.,
  $m\notin \eps(x,y)$. Sloppy speaking, the phylogenetic signal to infer
  such trees is entirely contained in the \emph{non}-HGT events.

\section*{Acknowledgments}
We thank Manuela Gei{\ss} for all the stimulating discussions during the
plenty of interesting beer-sessions. Moreover, we thank Carmen Bruckmann
as well as the three anonymous referees for carefully rechecking the
established results which significantly helped to improve this paper.  This
work was supported in part by the German Federal Ministry of Education and
Research (BMBF, project no.\ 031A538A, de.NBI-RBC).

\bibliographystyle{plainnat}
\bibliography{multicol}

\end{document}